\documentclass[11pt]{article}
\usepackage{amsfonts}
\usepackage{latexsym,amssymb,amsmath,graphics,cite,arydshln}
\usepackage{tikz}
\usepackage{fancyhdr}
\usepackage{lipsum}
\usepackage{lineno}
\usepackage{color}

\usepackage{algorithm} 
\usepackage{algorithmic} 
\usepackage{multirow} 
\usepackage{amsmath}
\usepackage{xcolor}

\begin{document}
\newcommand{\qed}{\hphantom{.}\hfill $\Box$\medbreak}
\newcommand{\proof}{\noindent{\bf Proof \ }}

\newtheorem{theorem}{Theorem}[section]
\newtheorem{lemma}[theorem]{Lemma}
\newtheorem{corollary}[theorem]{Corollary}
\newtheorem{remark}[theorem]{Remark}
\newtheorem{example}[theorem]{Example}
\newtheorem{definition}[theorem]{Definition}
\newtheorem{construction}[theorem]{Construction}
\newtheorem{fact}[theorem]{Fact}
\newtheorem{proposition}[theorem]{Proposition}
\newtheorem{conjecture}[theorem]{Conjecture}

	\title{Construction of optimal flag codes by MRD codes\footnote{the National Natural Science Foundation of China under Grant $12101440$, the Natural Science Foundation of Jiangsu Province under Grant BK$20210858$ (Liu), and the National Natural Science Foundation of China under Grant $12271390$ (Ji). (Corresponding author: Shuhui Yu.)}}
	
	\author{\small Shuangqing Liu, Shuhui Yu, Lijun\ Ji \smallskip \footnote{	S. Liu is with the Department of Mathematics, Suzhou University of Science and Technology, Suzhou 215009, P. R. China  (e-mail: shuangqingliu@usts.edu.cn); S. Yu and L. Ji is with the Department of Mathematics, Soochow University, Suzhou 215006, P. R. China  (e-mail: yushuhui\_suda@163.com, jilijun@suda.edu.cn)}}
	\date{}
	\maketitle
	\begin{abstract}	
		Flag codes have received a lot of attention due to its application in random network coding. In 2021, Alonso-Gonz\'{a}lez et al. constructed optimal $(n,\mathcal{A})$-Optimum distance flag codes(ODFC) for $\mathcal {A}\subseteq \{1,2,\ldots,k,n-k,\ldots,n-1\}$ with $k\in \mathcal A$ and $k\mid n$.  
In this paper, we introduce a new construction of $(n,\mathcal A)_q$-ODFCs by maximum rank-metric codes, and prove that there is an $(n,\mathcal{A})$-ODFC of size $\frac{q^n-q^{k+r}}{q^k-1}+1$ for any $\mathcal{A}\subseteq\{1,2,\ldots,k,n-k,\ldots,n-1\}$ with
$\mathcal A\cap \{k,n-k\}\neq\emptyset$,
where $r\equiv n\pmod k$ and $0\leq r<k$.
Furthermore, when $k>\frac{q^r-1}{q-1}$, this $(n,\mathcal A)_q$-ODFC is optimal. Specially, when $r=0$, Alonso-Gonz\'{a}lez et al.'s result is also obtained. We also gives a characterization of almost optimum distance flag codes, and construct a family of optimal almost optimum flag distance codes.
		
		\smallskip
	\end{abstract}
	
	\section{Introduction}
	
	Random network coding, introduced in \cite{acly}, is a new method for attaining a maximum
	information flow by using a channel modelled as an acyclic-directed multigraph with
	possibly several senders and receivers. In this model, intermediate nodes are allowed to
	perform random linear combinations of the received vectors instead of simply routing
	them, as it happens when using classical channels of communication. It was also proved in \cite{acly} that the information rate
	of a network can be improved by the model. However, This process is specially vulnerable to error dissemination. To deal with this  problem,
	K\"{o}etter and Kschischang \cite{kk} introduced the concept of subspace codes as adequate error-correction codes in random network coding.
	
	Subspace codes are applied in a single use of the channel.
	When the times of use the subspace channel are more than once, we talk about multishot subspace
	codes. In this kind of codes, introduced in \cite{nu}, the subspace channel is used many times,
	in order to transmit sequences of subspaces. 
	As it was explained in \cite{nu}, multishot
	subspace codes appear as an interesting alternative to subspace codes (one-shot subspace
	codes) when the field size $q$ or the packet size $n$ cannot be increased. The class of flag
	codes appear as a particular case of multishot subspace codes. 
	
	Flag codes in network coding appeared for the first time in \cite{lnv}. Later on, some studies of flag codes attaining the best possible distance (optimum distance flag codes (ODFCs)) are undertaken. Alonso-Gonz\'{a}lez, Navarro-P\'{e}rez and Soler-Escriv\`{a} in \cite{ans} constructed optimum distance full flag codes with $n=2k$ by planar spreads. The authors in \cite{ans1} gave the constructions of ODFCs with $k| n$  by  perfect matchings in graphs.  Recently, Chen and Yao \cite{cy} presented  constructions of totally isotropic flag orbit codes with optimum distance on the symplectic space $\mathbb{F}_q^{2v}$. Kurz \cite{kurz} also presented some upper and lower bounds of flag codes, and constructed optimum distance full flag codes with $n=2k+1$. Recently, Alonso-Gonz\'{a}lez et al. \cite{ans2} introduced distance vectors to derive bounds for their maximum possible size once the
	minimum distance and dimensions are fixed. 
	
	This paper is devoted to constructing flag codes based on maximum rank-metric codes. It is proved that there is an $(n,\mathcal{A})$-ODFC of size
$\frac{q^n-q^{k+r}}{q^k-1}+1$ for any $\mathcal{A}\subseteq\{1,2,\ldots,k,n-k,\ldots,n-1\}$ with
$\mathcal A\cap \{k,n-k\}\neq\emptyset$,
where $r\equiv n\pmod k$ and $0\leq r<k$.
Furthermore, when $k>\frac{q^r-1}{q-1}$, this $(n,\mathcal A)_q$-ODFC is optimal. Specially, the case $r=0$ also gives Alonso-Gonz\'{a}lez et al.'s optimal ODFCs with $k\mid n$ in \cite{ans1}, while our constructions do not depend on the condition that $k \mid n$. We also gives a characterization of $(n,\{1,2,\ldots, k,n-k+1,\ldots,n-1\})$-almost optimum distance flag codes (AODFCs), and prove that there is an $(n,\{1,2,\ldots, k,n-k+1,\ldots,n-1\})$-AODFC of size $\frac{q^n-q^{k+r-1}}{q^{k-1}-1}+1$, where $r \equiv n \pmod{(k-1)}$ and $0\leq r<k-1$. Furthermore, when $k>\frac{q^r-1}{q-1}+1$, this $(n,\{1,2,\ldots, k,n-k+1,\ldots,n-1\})_q$-AODFC is optimal. 
	
	The rest of this paper is organized as follows.
	Section 2 gives a brief introduction of constant dimension codes (CDCs), flag codes and optimum distance flag codes. 	
	In Section 3, we give a construction for $(n,\mathcal A)_q$-ODFCs based on maximum rank-metric (MRD) codes (see Theorem \ref{conc-1}), where $\mathcal A\subseteq \{1,2,\ldots,n-1\}$.
We show that there is an $(n,\mathcal{A})$-ODFC of size
$\frac{q^n-q^{k+r}}{q^k-1}+1$ for $\mathcal{A}\subseteq\{1,2,\ldots,k,n-k,\ldots,n-1\}$ with
$\mathcal A\cap \{k,n-k\}\neq\emptyset$,
where $r\equiv n\pmod k$ and $0\leq r<k$.
Furthermore, when $k>\frac{q^r-1}{q-1}$, the ODFC is optimal (see Corollary \ref{rek}).  

In Section 4, we give a characterization of $(n,\{1,2,\ldots, k,n-k+1,\ldots,n-1\})$-AODFCs, and give a construction for $(n,\{1,2,\ldots, k,n-k+1,\ldots,n-1\})$-AODFCs based on maximum rank-metric (MRD) codes.
We show that there is an $(n,\{1,2,\ldots, k,n-k+1,\ldots,n-1\})$-AODFC of size $\frac{q^n-q^{k+r-1}}{q^{k-1}-1}+1$, where $r \equiv n \pmod{(k-1)}$ and $0\leq r<k-1$. Furthermore, when $k>\frac{q^r-1}{q-1}+1$, the AODFC is optimal (see Corollary \ref{rek2}).  
	
	\section{Preliminaries}
	
	Let $\mathbb F_q$ be the finite field of order $q$, and $\mathbb F_q^n$ the $n$-dimensional vector space over $\mathbb F_q$. Let ${\cal P}_q(n)$ denote the set of all subspaces of $\mathbb F_q^n$. Given a nonnegative integer $k\leq n$, the set of all $k$-dimensional subspaces of $\mathbb F_q^n$ is called the {\em Grassmannian} ${\cal G}_q(n,k)$.

	\subsection{Constant dimension codes}
	A subset $\cal C$ of ${\cal G}_q(n,k)$ is called an $(n,d,k)_q$ {\em constant dimension code} (CDC), if $\cal C$ satisfies the {\em subspace distance}
	\begin{eqnarray}\label{1.1}
		d_S(\mathcal U,\mathcal V)\triangleq{\rm dim}~\mathcal U+{\rm dim}~\mathcal V-2{\rm dim}~(\mathcal U\cap \mathcal V)\geq d
	\end{eqnarray}
	for any two subspaces $\mathcal U,\mathcal V\in \mathcal C$.
	Elements in $\cal C$ are called {\em codewords}.
	An $(n,d,k)_q$-CDC with $M$ codewords is written as an $(n,M,d,k)_q$-CDC.
	Given $n,d,k$ and $q$, denote by $A^C_q(n, d, k)$ the maximum number of codewords among all $(n,d,k)_q$-CDCs. 

	Clearly, the minimum distance of a constant dimension code $\mathcal C$ is upper-bounded by:
	\begin{eqnarray}
		d_S(\mathcal C)\leq\left \{
		\begin {aligned}
		&2k\qquad \quad {\rm if~} 2k\leq n\\
		&2(n-k)~~{\rm otherwise}.\\
		\end {aligned}
		\right.\nonumber
	\end{eqnarray}
	Specially, every $(n,2k,k)_q$-CDC is a partial $k$-spread in $\mathcal G(n,k)$.
	\begin{theorem}{\rm \cite{ns,df}}\label{cdc-upper}
		Let $n\equiv r\pmod k$, $0\leq r<k$. If $r\neq 0$ then
		\begin{center}
			$A^C_q(n, 2k, k)\leq \frac{q^n-q^r}{q^k-1}-\lfloor\frac{\sqrt{4q^k(q^k-q^r)+1}-(2q^k-2q^r+1)}{2}\rfloor-1$.
		\end{center}
		If $k>\begin{bmatrix}
			r \\
			1 \\
		\end{bmatrix}_q=\frac{q^r-1}{q-1}$, then $A^C_q(n, 2k, k)= \frac{q^n-q^{k+r}}{q^k-1}+1$.
	\end{theorem}
	
	A $k$-dimensional subspace $\mathcal U$ of $\mathbb F^n_q$ can be represented by a $k\times n$ generator matrix whose rows form a basis for $\mathcal U$.
	If $\mathcal U,\mathcal V\in {\cal G}_q(n,k)$, the subspace distance on ${\cal G}_q(n,k)$ is also given by
	\begin{eqnarray}\label{1.2}
		d_S(\mathcal U,\mathcal V)=2\cdot {\rm rank} \binom{\boldsymbol{U}}{\boldsymbol{V}}-2k, \nonumber
	\end{eqnarray}
	where $\boldsymbol{U},\boldsymbol{V} \in \mathbb F^{k\times n}_q$ are matrices such that $\mathcal U$ =rowspace$(\boldsymbol{U})$ and $\mathcal V$ =rowspace$(\boldsymbol{V})$, $\mathbb F^{k \times n}_q$ is the set of all $k \times n$ matrices over $\mathbb F_q$. Clearly, generator matrices are usually not unique.
	However there exists a unique matrix representation for each element of the Grassmannian, namely the {\em reduced row echelon form}({\em inverse reduced row echelon form}). A $k\times n$ matrix with rank $k$ is in reduced row echelon form (resp. inverse reduced row echelon form) if
	$(1)$ the leading coefficient of a row is always to the right (resp. left) of the leading coefficient of the previous row;
	$(2)$ all leading coefficients are ones;
	$(3)$ every leading coefficient is the only nonzero entry in its column.
	
	For a subspace $\mathcal U\in \mathcal G_q(n,k)$, $\mathcal U$ can be represented by a $k\times n$ matrix $\boldsymbol U$ in reduced row echelon form(inverse reduced row echelon form), whose rows form a basis for $\mathcal U$.
	The {\em identifying vector}({\em inverse identifying vector}) $\boldsymbol v(\mathcal U)$($\boldsymbol{\hat{v}(\mathcal U)}$) or $\boldsymbol v(\boldsymbol U)$($\boldsymbol{\hat{v}(\boldsymbol U)}$) is the binary vector of length $n$ and weight $k$ such that the $1's$ of $\boldsymbol v(\boldsymbol U)$($\boldsymbol{\hat{v}(\boldsymbol U)}$) are in the positions where $\boldsymbol U$ has its leading ones. 
	
	\begin{lemma}{\rm \cite{es09}}\label{lem:efc-2}
		Let $\mathcal U,\mathcal V \in \mathcal G_q(n,k)$, $\mathcal U=rowspace(\boldsymbol U)$ and $\mathcal V=rowspace(\boldsymbol V)$, where $\boldsymbol U, \boldsymbol V\in \mathbb F_q^{k\times n}$ are in reduced row echelon forms. Then 
		\begin{center}
			$d_S(\mathcal U, \mathcal V) \geq d_H(\boldsymbol{v(U)}, \boldsymbol{v(V)})$,
		\end{center}
		where $d_H(\boldsymbol{v(U)}, \boldsymbol{v(V)})$ is Hamming distance of $\boldsymbol{v(U)}$ and $\boldsymbol{v(V)}$. 
	\end{lemma}
	
	\begin{lemma}{\rm \cite{lj}}\label{lem:iefc}
		Let $\mathcal U,\mathcal V \in \mathcal G_q(n,k)$, $\mathcal U=rowspace(\boldsymbol U)$ and $\mathcal V=rowspace(\boldsymbol V)$, where $\boldsymbol U, \boldsymbol V\in \mathbb F_q^{k\times n}$ are in inverse reduced row echelon forms. Then 
		\begin{center}
			$d_S(\mathcal U, \mathcal V) \geq d_H(\boldsymbol{\hat{v}(U)}, \boldsymbol{\hat{v}(V)})$.
		\end{center}
	\end{lemma}
	
	Rank-metric codes, especially MRD codes, play an important role in the constructions of CDCs. These constructions of CDCs can be seen in \cite{es09,xf,lj,lj1,lcf1} and http://subspace
	codes.uni-bayreuth.de.
	
	For a matrix $\boldsymbol{A} \in \mathbb F^{m \times n}_q$, the rank of $\boldsymbol{A}$ is denoted by rank$(\boldsymbol{A})$. The {\em rank distance} on $\mathbb F^{m \times n}_q$ is defined by
	\begin{center}
		$d_R(\boldsymbol{A}, \boldsymbol{B})={\rm rank}(\boldsymbol{A}-\boldsymbol{B}),~{\rm for}~\boldsymbol{A}, \boldsymbol{B} \in \mathbb F^{m \times n}_q$.
	\end{center}
	An $[m \times n, k, \delta]_q$ {\em rank-metric code} $\mathcal C$ is a $k$-dimensional $\mathbb F_q$-linear subspace of $\mathbb F^{m \times n}_q$ with minimum rank distance $\delta$.
	The Singleton-like upper bound for rank-metric codes implies that $$k \leq {\rm max}\{m,n\}({\rm min}\{m,n\}- \delta +1)$$
	holds for any $[m \times n, k, \delta]_q$ code. When the equality holds, $\cal C$ is called a {\em linear maximum rank-metric code}, denoted by an $[m \times n, \delta]_q$-MRD code. Linear MRD codes exist for all feasible parameters (cf. \cite{d,g,r}).
	
	\begin{theorem}{\rm\cite{d}}\label{MRD}
		Let $m,n,\delta$ be positive integers. Then there is an $[m \times n, \delta]_q$-MRD code, i.e, an $[m \times n, k, \delta]_q$ rank-metric code with $k={\rm max}\{m,n\}({\rm min}\{m,n\}- \delta +1)$. 
	\end{theorem}
	
	\begin{lemma}{\rm\cite{es09}}\label{lem:efc-1}
		Let $\mathcal U, \mathcal{V}\in \mathcal G_q(n,k)$, $\mathcal U=rowspace(\boldsymbol U)$ and $\mathcal V=rowspace(\boldsymbol V)$, where $\boldsymbol U, \boldsymbol V\in \mathbb F_q^{k\times n}$ are in reduced row echelon forms.
		If $\boldsymbol{v(U)} = \boldsymbol{v(V)}$, then
		\begin{center}
			$d_S(\mathcal U,\mathcal V)=2d_R(\boldsymbol C_{\boldsymbol U}, \boldsymbol C_{\boldsymbol V})$,
		\end{center}
		where $\boldsymbol C_{\boldsymbol U}$ and $\boldsymbol C_{\boldsymbol V}$ denote the submatrices of $\boldsymbol{U}$ and $\boldsymbol{V}$ without the columns of their pivots, respectively.
	\end{lemma}
	
	\begin{lemma}{\rm\cite{lj}}\label{lem:iefc-1}
		Let $\mathcal U, \mathcal{V}\in \mathcal G_q(n,k)$, $\mathcal U=rowspace(\boldsymbol U)$ and $\mathcal V=rowspace(\boldsymbol V)$, where $\boldsymbol U, \boldsymbol V\in \mathbb F_q^{k\times n}$ are in inverse reduced row echelon forms.
		If $\boldsymbol{\hat{v}(U)} = \boldsymbol{\hat{v}(V)}$, then
		\begin{center}
			$d_S(\mathcal U,\mathcal V)=2d_R(\boldsymbol C_{\boldsymbol U}, \boldsymbol C_{\boldsymbol V})$,
		\end{center}
		where $\boldsymbol C_{\boldsymbol U}$ and $\boldsymbol C_{\boldsymbol V}$ denote the submatrices of $\boldsymbol{U}$ and $\boldsymbol{V}$ without the columns of their pivots, respectively.
	\end{lemma}
	
	For more details on CDCs and rank-metric codes, see \cite{es09,BHLPRW2022,lj1} and the references therein.

	\subsection{Flag codes}
	A {\em flag} of type $(t_1,t_2,\ldots,t_r)$, with $0<t_1<t_2<\ldots<t_r<n$, on the vector space $\mathbb F_q^n$ is a sequence of subspaces $\mathcal F=(\mathcal F_1,\mathcal F_2,\ldots,\mathcal F_r)$ in
	${\cal G}_q(n,t_1)\times {\cal G}_q(n,t_2)\times\cdots\times{\cal G}_q(n,t_r)\subseteq {\cal P}_q(n)^r$ such that
	\begin{center}
		$\{0\}\subsetneq \mathcal F_1\subsetneq \cdots\subsetneq \mathcal F_r\subsetneq \mathbb F_q^n$.
	\end{center}
	With this notation, $\mathcal F_i$ is said to be the $i$-th subspace of $\mathcal F$. In case the type vector is $(1,2,\ldots,n-1)$ we say that $\mathcal F$ is a {\em full flag}.
	
	The space of flags of type $(t_1,t_2,\ldots,t_r)$ on $\mathbb F_q^n$ is denoted by $\mathcal F_q(n, (t_1,t_2,\ldots,t_r))$ and
	can be endowed with the flag distance $d_f$ that naturally extends the subspace distance defined in (\ref{1.1}): given two flags $\mathcal F=(\mathcal F_1,\mathcal F_2,\ldots,\mathcal F_r)$ and
	$\mathcal F'=(\mathcal F'_1,\mathcal F'_2,\ldots,\mathcal F'_r)$ in $\mathcal F_q(n, (t_1,t_2,\ldots,t_r))$, the flag distance between them is
	\begin{center}
		$d_f(\mathcal F,\mathcal F')=\sum\limits_{i=1}^{r}d_S(\mathcal F_i,\mathcal F'_i)$.
	\end{center}
	
	A flag code of type $(t_1,t_2,\ldots,t_r)$ on $\mathbb F_q^n$ is defined as any non-empty subset $\mathcal C \subseteq \mathcal F_q(n, (t_1,t_2,\ldots,t_r))$. The minimum distance of a flag code $\mathcal C$ of type
	$(t_1,t_2,\ldots,t_r)$ on $\mathbb F_q^n$ is given by
	\begin{center}
		$d_f(\mathcal C)=\min\{d_f(\mathcal F,\mathcal F'):\mathcal F,\mathcal F'\in \cal C,~\mathcal F\neq\mathcal F'\}$.
	\end{center}
	\begin{lemma}{\rm \cite{ans,ans1}}\label{upper}
		Given a flag code $\mathcal C$ of type $(t_1,t_2,\ldots,t_r)$, its minimum distance is upper-bounded by:
		\begin{eqnarray}\label{disf}
			d_f(\mathcal C)\leq
			2\left(\sum\limits_{t_i\leq \lfloor\frac{n}{2}\rfloor}t_i+\sum\limits_{t_i> \lfloor\frac{n}{2}\rfloor}(n-t_i)\right).
		\end{eqnarray}
	\end{lemma}
	
	A flag code $\mathcal C$ of type $(t_1,t_2,\ldots,t_r)$ is called an $(n,(t_1,t_2,\ldots,t_r))_q$ {\em optimum distance flag code} (ODFC) if  the distance attains $d_{max}$ for flag codes of type $(t_1,t_2,\ldots,t_r)$ in Lemma \ref{upper}. A flag code $\mathcal C$ of type $(t_1,t_2,\ldots,t_r)$ is called an $(n,(t_1,t_2,\ldots,t_r))_q$ {\em almost optimum distance flag code} (AODFC) if  the minimum distance is $d_{max}-2$, where $d_{max}$ is the upper bound of type $(t_1,t_2,\ldots,t_r)$ in Lemma \ref{upper}.
 Let $A_q(n, \{t_1,t_2,\ldots,t_r\})$($A^{*}_q(n, \{t_1,t_2,\\ \ldots,t_r\})$) denote the maximum number of codewords among all $(n,(t_1,t_2,\ldots,t_r))_q$-ODFCs\\(AODFCs).
An $(n,(t_1,t_2,\ldots,t_r))_q$-ODFC(AODFC) with $A_q(n, \{t_1,t_2,\ldots,t_r\})$ $(A^{*}_q(n, \{t_1,\\t_2,\ldots,t_r\}))$ codewords is said to be {\em optimal}.

	Given a type vector $(t_1,t_2,\ldots,t_r)$, for every $i=1,2,\ldots,r$, we define the {\em $i$-projection} to be the map
	\begin{center}
		$p_i: \mathcal F_q(n, (t_1,t_2,\ldots,t_r))\longrightarrow {\cal G}_q(n,t_i)$
	\end{center}
	\begin{center}
		$\mathcal F=(\mathcal F_1,\mathcal F_2,\ldots,\mathcal F_r)\longmapsto p_i(\mathcal F)=\mathcal F_i$.
	\end{center}
	The {\em $i$-projected code} of $\mathcal C$ is the set $\mathcal C_i=\{p_i(\mathcal F):\mathcal F\in \mathcal C\}$. By definition, this code $\mathcal C_i$ is a constant dimension code in the ${\cal G}_q(n,t_i)$
	and its cardinality satisfies $|\mathcal C_i|\leq|\mathcal C|$. We say that $\mathcal C$ is a {\em disjoint flag code} if $|\mathcal C_1|=|\mathcal C_2|=\cdots=|\mathcal C_r|=|\mathcal C|$,  that is, the
	$i$-projection $p_i$ is injective for any $i\in\{1,2,\ldots,r\}$.
	
	\begin{theorem}{\rm \cite{ans}}\label{ODFC-Proj}
		Let $\mathcal C$ be a flag code of type $(t_1,t_2,\ldots,t_r)$ on $\mathbb F_q^n$. The following statements are equivalent:
		\begin{itemize}
			\item[$(1)$] The code $\mathcal C$ is an optimum distance flag code.
			\item[$(2)$] The code $\mathcal C$ is disjoint and every projected code $\mathcal C_i$ attains the maximum possible distance, that is, $d_S(\mathcal C_i)=\min\{2t_i,2(n-t_i)\}$.
		\end{itemize}
	\end{theorem}

\begin{remark}\label{ODFCrek}
Condition $(2)$ in Theorem \ref{ODFC-Proj} can be simplified as: the code $\mathcal C$ is disjoint and  $d_S(\mathcal C_a)=\min\{2t_a,2(n-t_a)\}, d_S(\mathcal C_b)=\min\{2t_b,2(n-t_b)\}$, where $t_{a}$ and $t_{b}$ are the largest $t_i$ that are less than or equal to $\frac{n}{2}$ and the smallest $t_i$ that are greater than or equal to  $\frac{n}{2}$, respectively.
\end{remark}

	As a consequence, all the $i$-projected codes of an optimum distance flag code have to be $(n,2t_i,t_i)_q$-CDCs if $t_i\leq \lfloor\frac{n}{2}\rfloor$
	and $(n,2n-t_i,t_i)_q$-CDCs if $t_i> \lfloor\frac{n}{2}\rfloor$. By Theorem \ref{cdc-upper}, we have the following result.
	
	\begin{theorem}\label{flagupper}
		Let $n,k$ be positive integers with $\frac{n}{k}\geq 2$. Let $\mathcal{A}\subseteq\{1,2,\ldots,k,n-k,\ldots,n-1\}$.
		Suppose that $\mathcal A\cap \{k,n-k\}\neq\emptyset$ and
		set $n\equiv r\pmod k$, $0\leq r<k$. If $r\neq 0$, then 
		\begin{center}
			$A_q(n,\mathcal A)\leq \frac{q^n-q^r}{q^k-1}-\lfloor\frac{\sqrt{4q^k(q^k-q^r)+1}-(2q^k-2q^r+1)}{2}\rfloor-1$.
		\end{center}
		If $k>\begin{bmatrix}
			r \\
			1 \\
		\end{bmatrix}_q$, then $A_q(n,\mathcal A)\leq \frac{q^n-q^{k+r}}{q^k-1}+1$.
	\end{theorem}

	When $k$ is a divisor of $n$, Alonso-Gonz\'{a}lez, Navarro-P\'{e}rez and Soler-Escriv\`{a} showed the following result.
	\begin{theorem}\label{dn}{\rm \cite{ans1}}
		Let $k$  be a divisor of $n$, and $\{t_1,t_2,\ldots,t_r\}\subseteq\{1,2,\ldots,k,n-k,\ldots,n-1\}$.
		If $k\in \{t_1,t_2,\ldots,t_r\}$,
		then
		\begin{center}
			$A_q(n, \{t_1,t_2,\ldots,t_r\})=\frac{q^n-1}{q^k-1}$.
		\end{center}
	\end{theorem}
	
	When $n=2k+1$, Sascha Kurz showed the following result by using $k$-partial spread of $\mathbb{F}_q^{2k+1}$. 
	\begin{theorem}\label{nd}{\rm \cite{kurz}}
		For each integer $k\geq 2$, it holds that 
		$A_q(2k+1, \{1,2,\ldots,2k\})=q^{k+1}+1$.
	\end{theorem}

	\section{Multilevel construction based on MRD codes}
	
	The multilevel construction  \cite{es09} and inverse multilevel construction \cite{lj} for constant dimension codes are simple but effective, which generalize the lifted MRD codes. 
	In this section, we use  multilevel construction and inverse multilevel construction \cite{lj} to construct $(n,\mathcal A)_q$-ODFCs, where $\mathcal A\subseteq \{1,2,\ldots,n-1\}$, by MRD codes.

	\begin{theorem}\label{conc-1}
		Let $n,k$ be positive integers with $n\geq 2k$. Let $\mathcal{A}\subseteq\{1,2,\ldots,k,n-k,\ldots,n-1\}$.
		If $\mathcal A\cap \{k,n-k\}\neq\emptyset$,  then
		\begin{align*}
			A_q(n,\mathcal A)\geq\frac{q^n-q^{k+r}}{q^k-1}+1,
		\end{align*}
		where $r\equiv n\pmod k$ and $0\leq r<k$.
	\end{theorem}
	
	\begin{proof} Set $n=(a+1)k+r$.
		
		{\bf Step 1.} We first construct the set $\mathcal{C}$ of sequences of subspaces as follows.
		
		Let $\mathcal{D}_{a-1}$ be an $[(n-k)\times(n-k),n-k]_q$-MRD code, which exists by Theorem \ref{MRD}. For each nonzero matrix $D=\left(
		\begin{array}{c}
			\boldsymbol A \\
			\boldsymbol B \\
		\end{array}
		\right)\in\mathcal D_{a-1}$ with rank$(\boldsymbol A)=k$, define an invertible 
		matrix of order $n$ as follows:
		\begin{align*}
			M(D)=\left(
			\begin{array}{cc}
				\boldsymbol A & \boldsymbol I_{k} \\
				\boldsymbol B & \boldsymbol O \\
				\boldsymbol A & \boldsymbol O \\
			\end{array}
			\right).
		\end{align*}
		If $a>1$, let $\mathcal{D}_i$ be an  $[(n-(i+2)k)\times (n-(i+2)k),n-(i+2)k]_q$-MRD code for $0\leq i\leq a-2$,  which exists by Theorem \ref{MRD}. For each nonzero matrix $G_i=\left(\begin{array}{c}
			\boldsymbol C \\
			\boldsymbol D \\
		\end{array}
		\right)\in \mathcal{D}_i$ with rank$(\boldsymbol C)=k$, define
		\begin{align*}
			M(G_i)=
			\left(
			\begin{array}{cccc}
				\boldsymbol O & \boldsymbol I_k & \boldsymbol C & \boldsymbol O \\
				\boldsymbol O & \boldsymbol O & \boldsymbol D & \boldsymbol O \\
				\boldsymbol I_{ik} & \boldsymbol O & \boldsymbol O & \boldsymbol O \\	
				\boldsymbol O & \boldsymbol O & \boldsymbol O & \boldsymbol I_k \\
				\boldsymbol O & \boldsymbol O & \boldsymbol C & \boldsymbol O \\
			\end{array}
			\right).
		\end{align*}
		For the zero matrix $G_i=O\in \mathcal{D}_i$, define
		\begin{align*}
			M(O)=
			\left( 
			\begin{array}{ccccc}
				\boldsymbol O & \boldsymbol I_k & \boldsymbol O  & \boldsymbol O& \boldsymbol O \\
				\boldsymbol O & \boldsymbol O & \boldsymbol O & \boldsymbol I_{n-(i+3)k} & \boldsymbol O \\
				\boldsymbol I_{ik} & \boldsymbol O & \boldsymbol O  & \boldsymbol O & \boldsymbol O \\
				\boldsymbol O & \boldsymbol O &  \boldsymbol O &\boldsymbol O & \boldsymbol I_k \\
				\boldsymbol O & \boldsymbol O &  \boldsymbol I_{k} & \boldsymbol O & \boldsymbol O \\
			\end{array}
			\right),
		\end{align*}
		Define two invertible matrices as follows:
		\begin{align*}
			M(a)=\left(
			\begin{array}{ccc}
				\boldsymbol O & \boldsymbol O  & \boldsymbol I_k \\
				\boldsymbol O &\boldsymbol I_{n-2k} & \boldsymbol O \\ 
				\boldsymbol I_{k} &\boldsymbol O & \boldsymbol O \\
			\end{array}
			\right),
		\end{align*} 
		\begin{align*}
			M(a+1)=\left(
			\begin{array}{ccc}
				\boldsymbol O  & \boldsymbol I_{k} & \boldsymbol O \\
				\boldsymbol I_{n-2k} & \boldsymbol O & \boldsymbol O \\
				\boldsymbol O &\boldsymbol O & \boldsymbol I_{k} \\
			\end{array}
			\right).
		\end{align*}
		
		Let 
		\[
		\begin{array}{rl}
			\mathcal{M}=&\{M(D)\colon D\in \mathcal{D}_{a-1}, D\neq O\}\\
			& \cup\{M(G_i)\colon G_i\in \mathcal{D}_i,0\leq i\leq a-2\}\\
			& \cup \{M(a),M(a+1)\}.\\
		\end{array}
		\]
		It is easy to see that each matrix of $\mathcal{M}$ is an invertible $n\times n$ matrix. Note that $\mathcal{M}=\{M(D)\colon D\in \mathcal{D}_{a-1}, D\neq O\}\cup \{M(a),M(a+1)\}$ if $a=1$. 
		
		For each matrix $M\in \mathcal{M}$ and  $1\leq j\leq n$, let $\mathcal{F}(M_{[j]})=\text{rowspace}(M_{[j]})$,
		\[
		\begin{array}{l}
			\mathcal{F}(M)=(\mathcal{F}(M_{[1]}),\ldots,\mathcal{F}(M_{[k]}), \mathcal{F}(M_{[n-k]}),\ldots,\mathcal{F}(M_{[n-1]})),
		\end{array}
		\]
		where $M_{[j]}$ stands for  the top $j \times n$ submatrix of
$M$.  

We claim that $\mathcal{C}=\{\mathcal{F}(M)\colon M\in \mathcal{M}\}$ is the desired flag code. 
		
		{\bf Step 2.} We show that $\mathcal{C}$ is the set of flags of type $(1,2,\ldots,k,n-k,\ldots,n-1)$ and $|\mathcal{C}|=\frac{q^n-q^{k+r}}{q^k-1}+1$.
		
		Since each matrix $M$ of $\mathcal{M}$ is an invertible $n\times n$ matrix, we have $\dim (\mathcal{F}(M_{[j]}))=j$ for $j\in \{1,\ldots,k,n-k,\ldots,n-1\}$. Clearly, $\{0\}\subsetneq \mathcal{F}(M_{[1]})\subsetneq \mathcal{F}(M_{[2]})\subsetneq \cdots \subsetneq \mathcal{F}(M_{[k]}) \subsetneq  \mathcal{F}(M_{[n-k]}) \subsetneq \mathcal{F}(M_{[n-k+1]}) \subsetneq \cdots \subsetneq \mathcal{F}(M_{[n-1]})\subsetneq \mathbb{F}_q^n$. Therefore, $\mathcal{F}(M)$ is a flag of type $(1,\ldots,k,n-k,\ldots,n-1)$.
		
		Clearly, $|\mathcal{C}|=|\mathcal{M}|$. Since $\mathcal D_{a-1}$ is an $[(n-k)\times(n-k),n-k]_q$-MRD code and each $\mathcal{D}_i$ is an $[(n-(i+2)k)\times(n-(i+2)k),n-(i+2)k]_q$-MRD code, by the construction of $\mathcal{M}$ we have 
		\begin{align*}
			|\mathcal C|&=\sum\limits_{i=0}^{a-1}|\mathcal D_{i}|+1=q^{ak+r}+ \sum\limits_{i=1}^{a-1}{q^{ik+r}}+1\\ &=\sum\limits_{i=1}^{a}{q^{ik+r}}+1 =\frac{q^n-q^{k+r}}{q^k-1}+1.   
		\end{align*}
		
		{\bf Setp 3.} We prove that $\mathcal C$ is an optimum distance flag code. By Theorem \ref{ODFC-Proj} and Remark \ref{ODFCrek}, it suffices to show that the $t$-projected $\mathcal C^t$ of $\mathcal C$ attains the maximum possible distance for $t\in \{k,n-k\}$, where $\mathcal C^t=\bigcup\limits_{i=0}^{a+1}\mathcal C_i^t$, $\mathcal{C}_{a-1}^t=\{\mathcal{F}(M(D)_{[t]})\colon D\in \mathcal{D}_{a-1}, D\neq \boldsymbol O\}$, $\mathcal{C}_i^t=\{\mathcal{F}(M(G)_{[t]})\colon G\in \mathcal{D}_i\}$, $\mathcal{C}_{a}^t=\{\mathcal{F}(M(a)_{[t]})\}$, and $\mathcal{C}_{a+1}^t=\{\mathcal{F}(M(a+1)_{[t]})\}$.
We distinguish two cases:
		
		\begin{itemize}
			\item[$(1)$] $t=k$.
			
			\item$(1.1)$ We shall prove that $d_S(\mathcal{C}_i^k)\geq 2k$ for $0\leq i\leq a-1$. 
			
			For $0\leq i\leq a-2$, since $\mathcal{D}_{i}$ is an 
			$[(n-(i+2)k)\times (n-(i+2)k),n-(i+2)k]_q$-MRD code,
			$\{G_{[k]}\colon G\in \mathcal{D}_{i}\}$ is a $[k\times (n-(i+2)k),k]_q$-MRD code. For any two distinct subspaces $\mathcal U=\text{rowspace}(O_{k\times ik},I_k, G_{[k]},O)$, $\mathcal V=\text{rowspace}(O_{k\times ik},I_k, G'_{[k]},O) \in \mathcal C_{i}^k$, by Lemma \ref{lem:efc-1} we have $d_S(\mathcal U,\mathcal V)=2d_R(G_{[k]},G'_{[k]})= 2k$.

			Since $\mathcal{D}_{a-1}$ is an 
		$[(n-k)\times (n-k),n-k]_q$-MRD code,
		$\{D_{[k]}\colon D\in \mathcal{D}_{a-1}\}$ is a $[k\times (n-k),k]_q$-MRD code. For any two distinct subspaces $\mathcal U=\text{rowspace}(D_{[k]},I_k,O)$, $\mathcal V=\text{rowspace}(D'_{[k]},I_k,O) \in \mathcal C_{a-1}^k$, by Lemma \ref{lem:iefc-1} we have $d_S(\mathcal U,\mathcal V)=2d_R(D_{[k]},D'_{[k]})= 2k$. 
			
			\item$(1.2)$  We shall prove that $d_S(\mathcal U,\mathcal V)\geq 2k$ for any two subspaces $\mathcal U\in \mathcal C^k_{a}, \mathcal V \in \mathcal C_i^k$, where $0\leq i\leq a-1$. 
			
			Since the indentifying vector of $\mathcal U$ is  $(\underbrace{0\cdots0}_{n-k}\underbrace{1\cdots1}_{k})$ and the indentifying vector of $\mathcal V$ is $(\underbrace{0\cdots0}_{ik}\underbrace{1\cdots1}_{k}\underbrace{0\cdots0}_{n-(i+1)k})$ if $i<a-1$ or $(\underbrace{\alpha }_{n-k}\underbrace{0\cdots0}_{k})$ with $wt(\alpha)=k$ if $i=a-1$,  we have that $d_S(\mathcal U,\mathcal V)\geq 2k$ by Lemma\ref{lem:efc-2}.
			
			\item$(1.3)$  We shall prove that $d_S(\mathcal U,\mathcal V)\geq 2k$ for any two subspaces $\mathcal U\in \mathcal C^k_{a+1}, \mathcal V \in \mathcal C_i^k$, where $0\leq i\leq a$.
			
			Since both the identifying vector  and inverse identifying vector of $\mathcal U$ are  $(\underbrace{0\cdots0}_{n-2k}\underbrace{1\cdots1}_{k}\\\underbrace{0\cdots0}_{k})$, the identifying vector of $\mathcal V$ is $(\underbrace{0\cdots0}_{ik}\underbrace{1\cdots1}_{k}\underbrace{0\cdots0}_{n-(i+1)k})$ when $i\leq a-2$, the inverse identifying vector of $\mathcal V$ is $(\underbrace{0\cdots0}_{n-k}\underbrace{1\cdots1}_{k})$ when $i= a-1$, and both the identifying vector and  inverse identifying vector of $\mathcal V$ are $(\underbrace{0\cdots0}_{n-k}\underbrace{1\cdots1}_{k})$ when $i= a$, we have that $d_S(\mathcal U,\mathcal V)\geq 2k$ by Lemmas    \ref{lem:efc-2} and  \ref{lem:iefc}.

			\item$(1.4)$ We shall prove that $d_S(\mathcal U,\mathcal V)\geq 2k$ for any two subspaces $\mathcal U\in \mathcal C^k_{i}, \mathcal V \in \mathcal C_j^k$, where $0\leq i< j\leq a-2$.
			
		Since the indentifying vectors of $\mathcal U$ and $\mathcal V$ are $(\underbrace{0\cdots0}_{ik}\underbrace{1\cdots1}_{k}\underbrace{0\cdots0}_{n-(i+1)k})$ and $(\underbrace{0\cdots0}_{jk}\underbrace{1\cdots1}_{k}\\ \underbrace{0\cdots0}_{n-(j+1)k})$ respectively,
	we have that $d_S(\mathcal U,\mathcal V)\geq 2k$ by Lemma
	\ref{lem:efc-2}.
	\item$(1.5)$  We shall prove that $d_S(\mathcal U,\mathcal V)\geq 2k$ for any two subspaces $\mathcal U\in \mathcal C^k_{a-1}, \mathcal V \in \mathcal C^k_{i}$, where $0\leq i\leq a-2$. 

		Since the inverse indentifying vectors of $\mathcal U$ and $\mathcal V$ are $(\underbrace{0\cdots0}_{n-k}\underbrace{1\cdots1}_{k})$ and $(\underbrace{\alpha}_{n-k}\underbrace{0\cdots0}_{k})$ with $wt(\alpha)=k$ respectively,
we have that $d_S(\mathcal U,\mathcal V)\geq 2k$ by Lemma
\ref{lem:iefc}.
			
			\item[$(2)$]  $t=n-k$.
			
			\item$(2.1)$  We shall prove that $d_S(\mathcal{C}_i^t)\geq 2k$ for $0\leq i\leq a-1$.
			
			\item $(2.1.1)$ For any two distinct subspaces $\mathcal U, \mathcal V \in \mathcal C_{a-1}^{n-k}$, we have 
				\begin{align*}
				d_S(\mathcal U,\mathcal V)&=2{\rm rank}\left(
				\begin{array}{ccccc}
					\boldsymbol A & \boldsymbol I_k \\
					\boldsymbol B & \boldsymbol O\\
					\boldsymbol A' & \boldsymbol I_k \\
					\boldsymbol B' & \boldsymbol O\\
				\end{array}
				\right)-2(n-k)\\
				&=2{\rm rank}\left(
				\begin{array}{ccccc}
					\boldsymbol A-\boldsymbol A' & \boldsymbol O \\
					\boldsymbol B-\boldsymbol B' & \boldsymbol O\\
					\boldsymbol A' & \boldsymbol I_k \\
				\end{array}
				\right)-2(n-k)\\&=2(n-k+k-n+k)=2k,
			\end{align*}
			where the last equality holds because $\left(
			\begin{array}{c}
				\boldsymbol A \\
				\boldsymbol B \\
			\end{array}
			\right),\left (\begin{array}{c}
				\boldsymbol A' \\
				\boldsymbol B' \\
			\end{array}
			\right)$ are two distinct nonzero codewords of $[(n-k)\times (n-k),n-k]_q$-MRD code $\mathcal D$.

			\item$(2.1.2)$ For any two distinct subspaces $\mathcal U=\mathcal{F}(M(G)_{[n-k]}), \mathcal V=\mathcal{F}(M(G')_{[n-k]}) \in \mathcal C_i^{n-k}$ with $G\neq O, G'\neq O$ and $0\leq i<a-1$, similar to (2.1.1), it is easy to verify that $d_S(\mathcal U, \mathcal V )\geq 2k$.

			\item$(2.1.3)$ For any subspace $\mathcal U=\mathcal{F}(M(G)_{[n-k]}), \mathcal V=\mathcal{F}(M(O)_{[n-k]}) \in \mathcal C_i^{n-k}$ with $G\neq O$ and $0\leq i<a-1$, we have
			{\begin{align*}
					d_S(\mathcal U,\mathcal V)&=2{\rm rank}\left(
					\begin{array}{cccc}
						\boldsymbol I_{ik} & \boldsymbol O & \boldsymbol O & \boldsymbol O \\
						\boldsymbol O & \boxed{\boldsymbol I_k} & \boldsymbol C & \boldsymbol O \\
						\boldsymbol O & \boldsymbol O & \boldsymbol D & \boldsymbol O \\
						\boldsymbol O & \boldsymbol O & \boldsymbol O & \boldsymbol I_k \\
						\hdashline[2pt/2pt]
						\boldsymbol I_{ik} & \boldsymbol O & \boldsymbol O & \boldsymbol O \\
						\boldsymbol O & \boxed{\boldsymbol I_k} & \boldsymbol O & \boldsymbol O \\
						\boldsymbol O & \boldsymbol O & * & * \\
					\end{array}
					\right)-2(n-k)
\\&\geq 2{\rm rank}\left(
					\begin{array}{ccc;{2pt/2pt}c}
						\boldsymbol I_{ik} & \boldsymbol O & \boldsymbol O & \boldsymbol O \\
						\boldsymbol O & \boldsymbol I_k & \boldsymbol O & \boldsymbol O \\
						\boldsymbol O & \boldsymbol O  & \boldsymbol C & \boldsymbol O \\
						\boldsymbol O & \boldsymbol O & \boldsymbol D & \boldsymbol O \\\hdashline[2pt/2pt]
						\boldsymbol O & \boldsymbol O & \boldsymbol O & \boldsymbol I_k \\
					\end{array}
					\right)-2(n-k)\\&=2(ik+k+n-(i+2)k+k)-2(n-k)\\&=2k,
			\end{align*}}
			where $*$ stands for a submatrix.
			
			\item$(2.2)$ We shall prove that $d_S(\mathcal U,\mathcal V)\geq 2k$
	for any two subspaces $\mathcal U\in \mathcal C^{n-k}_{a}, \mathcal V \in \mathcal C_i^{n-k}$, where $0\leq i\leq a-1$.
			
			Since the indentifying vector of $\mathcal U$ is  $(\underbrace{0\cdots0}_{k}\underbrace{1\cdots1}_{n-k})$, the indentifying vector of $\mathcal V$ is of the form $(\underbrace{1\cdots1}_{(i+1)k}\underbrace{\alpha_2}_{n-(i+2)k}\underbrace{1\cdots1}_{k})$ for some $\alpha_2\in \mathbb F_2^{n-(i+2)k}$ of Hamming weight $n-(i+3)k$ when $i\leq a-2$, and the indentifying vector of $\mathcal V$ is of the form $(\underbrace{1\cdots1}_{n-k}\underbrace{0\cdots0}_{k})$ when $i=a-1$ we have that $d_S(\mathcal U,\mathcal V)\geq 2k$ by Lemma
	\ref{lem:efc-2}.
	
			\item$(2.3)$ We shall prove that $d_S(\mathcal U,\mathcal V)\geq 2k$ for any two subspaces $\mathcal U\in \mathcal C^{n-k}_{a+1}, \mathcal V \in \mathcal C_i^{n-k}$, where $0\leq i\leq a$.
			
		Since the inverse indentifying vector of $\mathcal U$ is  $(\underbrace{1\cdots1}_{n-k}\underbrace{0\cdots0}_{k})$ and the inverse indentifying vector of $\mathcal V$ is of the form $(\underbrace{\alpha}_{n-k}\underbrace{1\cdots1}_{k})$ for some $\alpha_2\in \mathbb F_2^{n-k}$ of Hamming weight $n-2k$,  we have that $d_S(\mathcal U,\mathcal V)\geq 2k$ by Lemma
	\ref{lem:iefc}.
			\item$(2.4)$ We shall prove that $d_S(\mathcal U,\mathcal V)\geq 2k$ for any two subspaces $\mathcal U\in \mathcal C^{n-k}_a$, $\mathcal V \in \mathcal C^{n-k}_{i}$, where $0\leq i\leq a-2$.

			\item$(2.4.1)$  Let $\mathcal U=\mathcal{F}(M(D)_{[n-k]})\in \mathcal C^{n-k}_{a-1}, \mathcal V=\mathcal{F}(M(G)_{[n-k]}) \in \mathcal C^{n-k}_{i}$ with $D$ and $G$ nonzero matrices. We have
			{ \small\begin{align*}
					d_S(\mathcal U,\mathcal V)&=2{\rm rank}\left(
					\begin{array}{c;{2pt/2pt}c}
						\boldsymbol A & \boxed{\boldsymbol I_k} \\
						\boldsymbol B & \boldsymbol O \\
						\hdashline[2pt/2pt]
						\begin{array}{ccc}
							\boldsymbol I_{ik} & \boldsymbol O & \boldsymbol O  \\
							\boldsymbol O & \boldsymbol I_k & \boldsymbol C  \\
							\boldsymbol O & \boldsymbol O & \boldsymbol D  \\
							\boldsymbol O & \boldsymbol O & \boldsymbol O  \\
						\end{array}
						& \begin{array}{c}
							\boldsymbol O \\
							\boldsymbol O \\
							\boldsymbol O \\
							\boxed{\boldsymbol I_k} \\
						\end{array} \end{array}
					\right)-2(n-k+1)\\&=2{\rm rank}\left(
					\begin{array}{c;{2pt/2pt}c}
						\boldsymbol A & \boldsymbol O \\
						\boldsymbol B & \boldsymbol O \\
						\hdashline[2pt/2pt]
						\begin{array}{ccc}
							\boldsymbol I_{ik} & \boldsymbol O & \boldsymbol O  \\
							\boldsymbol O & \boldsymbol I_k & \boldsymbol C  \\
							\boldsymbol O & \boldsymbol O & \boldsymbol D  \\
							\boldsymbol O & \boldsymbol O & \boldsymbol O  \\
						\end{array}
						& \begin{array}{c}
							\boldsymbol O \\
							\boldsymbol O \\
							\boldsymbol O \\
							\boldsymbol I_k \\
						\end{array}
						\\
					\end{array}
					\right)-2(n-k+1)\\&
					\geq 2[{\rm rank}\left(
					\begin{array}{c}
						\boldsymbol A  \\
						\boldsymbol B  \\
					\end{array}
					\right)+k]-2(n-k+1)=2k-2.
			\end{align*}}
			
			\item$(2.4.2)$  Let $\mathcal U=\mathcal{F}(M(D)_{[n-k]})\in \mathcal C^{n-k}_{a-1}, \mathcal V=\mathcal{F}(M(O)_{[n-k]}) \in \mathcal C^{n-k}_{i}$ with  $D$ a nonzero matrix. We have
			{ \small\begin{align*}
					d_S(\mathcal U,\mathcal V)&=2{\rm rank}\left(
					\begin{array}{c;{2pt/2pt}c}
						\boldsymbol A & \boxed{\boldsymbol I_k} \\
						\boldsymbol B& \boldsymbol O \\
						\hdashline[2pt/2pt]
						\begin{array}{ccc}
							\boldsymbol I_{ik} & \boldsymbol O & \boldsymbol O  \\
							\boldsymbol O & \boldsymbol I_k & \boldsymbol O  \\
							\boldsymbol O & \boldsymbol O & *\\
							\boldsymbol O & \boldsymbol O & \boldsymbol O  \\
							\boldsymbol O & \boldsymbol O &  *
						\end{array}
						& \begin{array}{c}
							\boldsymbol O \\
							\boldsymbol O \\
							\boldsymbol O \\
							\boxed{\boldsymbol I_k} \\
						\end{array}
						\\
					\end{array}
					\right)-2t\\&=2{\rm rank}\left(
					\begin{array}{c;{2pt/2pt}c}
						\boldsymbol A & \boldsymbol O \\
						\boldsymbol B& \boldsymbol O \\
						\hdashline[2pt/2pt]
						\begin{array}{ccc}
							\boldsymbol I_{ik} & \boldsymbol O & \boldsymbol O  \\
							\boldsymbol O & \boldsymbol I_k & \boldsymbol O  \\
							\boldsymbol O & \boldsymbol O & * \\
							\boldsymbol O & \boldsymbol O & \boldsymbol O  \\
						\end{array}
						& \begin{array}{c}
							\boldsymbol O \\
							\boldsymbol O \\
							\boldsymbol O \\
							\boxed{\boldsymbol I_k} \\
						\end{array}
						\\
					\end{array}
					\right)-2(n-k)\\&
					\geq 2[{\rm rank}\left(
					\begin{array}{c}
						\boldsymbol A  \\
						\boldsymbol B  \\
					\end{array}
					\right)+k]-2(n-k)\\&=2k-2.
			\end{align*}}
			
			\item$(2.5)$ We shall prove that $d_S(\mathcal U,\mathcal V)\geq 2k$ for any two subspaces $\mathcal U\in \mathcal C^{n-k}_i$, $\mathcal V \in \mathcal C^{n-k}_{j}$, where $0\leq i<j\leq a-2$.

			\item$(2.5.1)$
		For any two subspaces $\mathcal U=\mathcal{F}(M(G)_{[n-k]})\in \mathcal C^{n-k}_{a-1}, \mathcal V=\mathcal{F}(M(G')_{[n-k]}) \in \mathcal C^{n-k}_{i}$ with $G\in \mathcal{D}_i$, $G'\in \mathcal{D}_{j}$ and $G\neq O$, we have
			{\begin{align*}
					d_S(\mathcal U,\mathcal V)&=2{\rm rank}\left(
					\begin{array}{cccc}
						\boldsymbol I_{ik} & \boldsymbol O & \boldsymbol O & \boldsymbol O \\
						\boldsymbol O & \boxed{\boldsymbol I_k} & \boldsymbol C & \boldsymbol O \\
						\boldsymbol O & \boldsymbol O & \boldsymbol D & \boldsymbol O \\
						\boldsymbol O & \boldsymbol O & \boldsymbol O & \boldsymbol I_k \\
						\hdashline[2pt/2pt]
						\boldsymbol I_{ik} & \boldsymbol O & \boldsymbol O & \boldsymbol O \\
						\boldsymbol O & \boxed{\boldsymbol I_k} & \boldsymbol O & \boldsymbol O \\
						\boldsymbol O & \boldsymbol O & * & * \\
					\end{array}
					\right)-2(n-k)
					\\&\geq 2{\rm rank}\left(
					\begin{array}{ccc;{2pt/2pt}c}
						\boldsymbol I_{ik} & \boldsymbol O & \boldsymbol O & \boldsymbol O \\
						\boldsymbol O & \boldsymbol I_k & \boldsymbol O & \boldsymbol O \\
						\boldsymbol O & \boldsymbol O  & \boldsymbol C & \boldsymbol O \\
						\boldsymbol O & \boldsymbol O & \boldsymbol D & \boldsymbol O \\\hdashline[2pt/2pt]
						\boldsymbol O & \boldsymbol O & \boldsymbol O & \boldsymbol I_k \\
					\end{array}
					\right)-2(n-k)\\&=2(ik+k+n-(i+2)k+k)-2(n-k)\\&=2k.
			\end{align*}}
			\item$(2.5.2)$ For  any subspace $\mathcal U=\mathcal{F}(M(O)_{[n-k]})\in \mathcal C^{n-k}_i$, $\mathcal V=\mathcal{F}(M(G)_{[n-k]}) \in \mathcal C^{n-k}_{j}$ with $G\in \mathcal{D}_{j}$, we have
			{\begin{align*}
					d_S(\mathcal U,\mathcal V)&=2{\rm rank}\left(
					\begin{array}{cccc}
						\boldsymbol I_{ik+k} &   \boldsymbol O & \boldsymbol O \\
						\boldsymbol O &  \boldsymbol O& \boldsymbol I_{n-(i+2)k} \\
						\boldsymbol O & \boldsymbol *  & \boldsymbol O \\
						\hdashline[2pt/2pt]
						\boldsymbol I_{ik+k} & \boldsymbol O  & \boldsymbol O \\
						\boldsymbol O &\boldsymbol I_{k} & \boldsymbol * &  \\
						\boldsymbol O & \boldsymbol O &\boldsymbol *   \\
					\end{array}
					\right)-2(n-k)
                             \\&\geq 2{\rm rank}\left(
					\begin{array}{ccccc}
						\boldsymbol I_{ik+k} & \boldsymbol O & \boldsymbol O \\
						\boldsymbol O & \boldsymbol I_k & \boldsymbol O  \\
						\boldsymbol O & \boldsymbol O  & \boldsymbol I_{n-(i+2)k}
					\end{array}
					\right)-2(n-k)\\&=2(ik+2k+n-(i+3)k+k)-2(n-k)\\&=2k.
			\end{align*}}
		\end{itemize}

		{\bf Step 4.} For $\mathcal{A}\subset \{1,2,\ldots,k,n-k,\ldots,n-1\}$ with $\mathcal{A}\cap \{k,n-k\}\neq \emptyset$, deleting the $t$-projected $\mathcal{C}^t$ for all $t\not\in \mathcal{A}$ yields an $(n,\mathcal {A})$-ODFC of size $\frac{q^n-q^{k+r}}{q^k-1}+1$.  This completes the proof.  \qed
	\end{proof}

	\begin{corollary}\label{rek}
		Let $n,k$ be positive integers with $n\geq 2k$. Let $\mathcal{A}\subseteq\{1,2,\ldots,k,n-k,\ldots,n-1\}$.
		If $\mathcal A\cap \{k,n-k\}\neq\emptyset$. If
		$k>\begin{bmatrix}
			r \\
			1 \\
		\end{bmatrix}_q$ where $n\equiv r\pmod k$ and $0\leq r<k$, then
		\begin{align*}
			A_q(n,\mathcal A)=\frac{q^n-q^{k+r}}{q^k-1}+1.
		\end{align*}
	\end{corollary}
	
	\begin{proof} Apply Theorems \ref{flagupper} and \ref{conc-1}.
	\end{proof}
	
	Furthermore, set $k=\varepsilon n$ in Corollary \ref{rek}, where $0<\varepsilon<1$, we have
	\begin{align*}
		A_q(n,\mathcal A)=\frac{q^n-q^{\varepsilon n+r}}{q^{\varepsilon n}-1}+1.
	\end{align*}
	
	
	\begin{remark}
	Note that when $k$ is a divisor of $n$, Theorem \ref{conc-1} can also give optimal ODFCs, as described in Theorem \ref{dn}. However the construction in \cite{ans1} is invalid for $k\nmid n$. Our construction does not depend on this condition that $k \mid n$. When $n=2k+1$, since $k>\begin{bmatrix}
		1\\
		1 \\
	\end{bmatrix}_q=1$, Theorem \ref{conc-1} can also give the optimal ODFCs described in Theorem \ref{nd}.
	\end{remark}


\section{Almost optimum distance flag codes}

We gives a characterization of $(n,\{1,2,\ldots,k,n-k+1,\ldots,n-1\})$-AODFCs in the next theorem.

\begin{theorem}\label{AODFC-Proj}
Let $\mathcal{C}$ be a flag code of type $(1,2,\ldots,k,n-k+1,\ldots,n-1)$. The following statements are equivalent:

$(i)$ The code $\mathcal{C}$  is an almost optimum distance flag code.

$(ii)$ The code $\mathcal{C}$ is disjoint, for any $i\notin \{k,n-k+1\}$, every projected code $\mathcal{C}_{i}$ attains the maximum possible subspace distance, $d_{S}(\mathcal{C}_{k})\geq 2k-2$, $d_{S}(\mathcal{C}_{n-k+1})\geq 2k-4$, and does not exist $\mathcal{F},\mathcal{F}'\in \mathcal{C}$ such that $d_{S}(\mathcal{F}_{k},\mathcal{F}'_{k})=2k-2$ and $d_{S}(\mathcal{F}_{n-k+1},\mathcal{F}'_{n-k+1})=2k-4$. 
\end{theorem}

\begin{proof}
$(i) \Rightarrow (ii)$. Let $\mathcal{C}$ be an almost optimum distance flag code. Firstly, we prove $\mathcal{C}$ is disjoint. 

Assume that $\mathcal{C}$ is not disjoint. Then, there exists some index $j\in \{1,2,\ldots,k,n-k+1,\ldots,n-1\}$ with $|\mathcal{C}_{j}|<|\mathcal{C}|$ and at least two different flags $\mathcal{F},\mathcal{F}'\in \mathcal{C}$ such that $\mathcal{F}_{j}=\mathcal{F}'_{j}$. If $j\leq n-2$, then $d_{f}(\mathcal{F},\mathcal{F}')=\Sigma_{i\neq j,i\neq n-1}d_{S}(\mathcal{F}_{i},\mathcal{F}'_{i})+d_{S}(\mathcal{F}_{j},\mathcal{F}'_{j})+d_{S}(\mathcal{F}_{n-1},\mathcal{F}'_{n-1})=\Sigma_{i\neq j,i\neq n-1}d_{S}(\mathcal{F}_{i},\mathcal{F}'_{i})+0+2(n-1-j)\leq 2d_{max}-4$, which contradicts the minimum distance of $\mathcal{C}$. If $j=n-1$, by the assumption of $\mathcal{C}$, we have $d_{S}(\mathcal{F}_{n-2},\mathcal{F}'_{n-2})=4$, then $\dim(\mathcal{F}_{n-1}+\mathcal{F}'_{n-1})\geq \dim(\mathcal{F}_{n-2}+\mathcal{F}'_{n-2})=n$, which contradicts $\mathcal{F}_{n-1}=\mathcal{F}'_{n-1}$. 

Secondly, we prove for any $i\notin \{k,n-k+1\}$, every projected code $\mathcal{C}_{i}$ attains the maximum possible subspace distance. 

Assume that $\mathcal{C}_{j}$ does not attains the maximum possible subspace distance, where $j\notin \{k,n-k+1\}$. Then, there exists some flags $\mathcal{F}$ and $\mathcal{F}'$ such that  $d_{S}(\mathcal{F}_{j},\mathcal{F}'_{j})<\{2j,2(n-j)\}$. If $j\leq k-1$, then $d_{S}(\mathcal{F}_{j},\mathcal{F}'_{j})<2j$, it follows that $d_{f}(\mathcal{F},\mathcal{F}')=\Sigma_{i\neq j,i\neq k}d_{S}(\mathcal{F}_{i},\mathcal{F}'_{i})+d_{S}(\mathcal{F}_{j},\mathcal{F}'_{j})+d_{S}(\mathcal{F}_{k},\mathcal{F}'_{k})\leq 2d_{max}-4$, which contradicts the minimum distance of $\mathcal{C}$. If $j\geq n-k+2$, then $d_{S}(\mathcal{F}_{j},\mathcal{F}'_{j})<2(n-j)$,  it follows that $d_{S}(\mathcal{F}_{n-k+1},\mathcal{F}'_{n-k+1})<2k-2$ and $d_{f}(\mathcal{F},\mathcal{F}')=\Sigma_{i\neq j,i\neq n-k+1}d_{S}(\mathcal{F}_{i},\mathcal{F}'_{i})+d_{S}(\mathcal{F}_{j},\mathcal{F}'_{j})+d_{S}(\mathcal{F}_{n-k+1},\mathcal{F}'_{n-k+1})\leq 2d_{max}-4$, which contradicts the minimum distance of $\mathcal{C}$. 

Finally, we prove that $d_{S}(\mathcal{C}_{k})\geq 2k-2$, $d_{S}(\mathcal{C}_{n-k+1})\geq 2k-4$, and there does not exist $\mathcal{F},\mathcal{F}'\in \mathcal{C}$ such that $d_{S}(\mathcal{F}_{k},\mathcal{F}'_{k})=2k-2$, $d_{S}(\mathcal{F}_{n-k+1},\mathcal{F}'_{n-k+1})=2k-4$.

Assume on the contrary that there  exist $\mathcal{F},\mathcal{F}'\in \mathcal{C}$ such that $d_{S}(\mathcal{F}_{k},\mathcal{F}'_{k})\leq 2k-4$ or $d_{S}(\mathcal{F}_{n-k+1},\mathcal{F}'_{n-k+1})\leq 2k-6$. Then $d_{f}(\mathcal{F},\mathcal{F}')=\Sigma_{i\neq k}d_{S}(\mathcal{F}_{i},\mathcal{F}'_{i})+d_{S}(\mathcal{F}_{k},\mathcal{F}'_{k})\leq 2d_{max}-4$ or $d_{f}(\mathcal{F},\mathcal{F}')=\Sigma_{i\neq n-k+1}d_{S}(\mathcal{F}_{i},\mathcal{F}'_{i})+d_{S}(\mathcal{F}_{n-k+1},\mathcal{F}'_{n-k+1})\leq 2d_{max}-4$, which contradicts the minimum distance of $\mathcal{C}$. 

Assume on the contrary that there  exist $\mathcal{F},\mathcal{F}'\in \mathcal{C}$ such that $d_{S}(\mathcal{F}_{k},\mathcal{F}'_{k})=2k-2$, $d_{S}(\mathcal{F}_{n-k+1},\mathcal{F}'_{n-k+1})=2k-4$. Then $d_{f}(\mathcal{F},\mathcal{F}')=\Sigma_{i\neq k,i\neq n-k+1}d_{S}(\mathcal{F}_{i},\mathcal{F}'_{i})+d_{S}(\mathcal{F}_{k},\mathcal{F}'_{k})+d_{S}(\mathcal{F}_{n-k+1},\mathcal{F}'_{n-k+1})\leq 2d_{max}-4$, which contradicts the minimum distance of $\mathcal{C}$. 

$(ii) \Rightarrow(i)$. Assume that $(ii)$ is true. Since $\mathcal{C}$ is disjoint, given any pair of different flags $\mathcal{F}$ and $\mathcal{F}'$ in $\mathcal{C}$, by assumption, $d_{f}(\mathcal{F},\mathcal{F}')=\Sigma_{i\notin [k,n-k+1]}+d_{S}(\mathcal{F}_{k},\mathcal{F}'_{k})+d_{S}(\mathcal{F}_{n-k+1},\mathcal{F}'_{n-k+1})\geq 2d_{max}-2$. Hence, $d_{f}(\mathcal{C})\geq 2d_{max}-2$ and $\mathcal{C}$ is an almost optimum distance flag code.
\end{proof}

\begin{remark}\label{AODFCrek}
Condition $(2)$ in Theorem \ref{AODFC-Proj} can be simplified as: the code $\mathcal C$ is disjoint and  $d_S(\mathcal C_{k-1})\geq 2k-2, d_S(\mathcal C_k)\geq 2k-2, d_S(\mathcal C_{n-k+1})\geq 2k-4, d_S(\mathcal C_{n-k+2})\geq 2k-4,$ and does not exists $\mathcal{F},\mathcal{F}'\in \mathcal{C}$ such that $d_{S}(\mathcal{F}_{k},\mathcal{F}'_{k})=2k-2$, $d_{S}(\mathcal{F}_{n-k+1},\mathcal{F}'_{n-k+1})=2k-4$. 
\end{remark}

As a consequence, the $i$-projected codes of an almost optimum distance flag code have to be an $(n,2i,i)_q$-CDC if $i\leq k-1$
and an $(n,2n-i,i)_q$-CDC if $i\geq n-k+2$. By Theorem \ref{cdc-upper}, we have the following result.

\begin{theorem}\label{flagupper2}
Let $n,k$ be positive integers with $n\geq 2k$. Then 
\begin{center}
$A_q^{*}(n,\{1,2,\ldots,k,n-k+1,\ldots,n-1\})\leq \frac{q^n-q^r}{q^{k-1}-1}-\lfloor\frac{\sqrt{4q^{k-1}(q^{k-1}-q^r)+1}-(2q^{k-1}-2q^r+1)}{2}\rfloor-1$,
\end{center}
where $r \equiv n \pmod{(k-1)}$ 
and $0\leq r<k-1$.
If $k>\begin{bmatrix}
 r \\
 1 \\
\end{bmatrix}_q+1$, then $A_q^{*}(n,\{1,2,\ldots,k,n-k+1,\ldots,n-1\})\leq \frac{q^n-q^{k+r-1}}{q^{k-1}-1}+1$.
\end{theorem}

In the rest of this section, we use multilevel construction to construct $(n,\{1,2,\ldots,k,n\\-k+1,\ldots,n-1\})_{q}$-AODFCs by MRD codes.

\begin{theorem}\label{conc-2}
	Let $n,k$ be positive integers with $n\geq 2k$. then
	\begin{align*}
		A^{*}_q(n,\{1,2,\ldots,k,n-k+1,\ldots,n-1\})\geq\frac{q^n-q^{k+r-1}}{q^{k-1}-1}+1,
	\end{align*}
	where $r \equiv n \pmod{(k-1)}$ 
and $0\leq r<k-1$.
\end{theorem}

\begin{proof} Set $n=(a+1)(k-1)+r$.
	
{\bf Step 1.} We first construct the set $\mathcal{C}$ of sequences of subspaces as follows.

Let $\mathcal{D}_{a-1}$ be an $[(n-k+1)\times(n-k+1),n-k+1]_q$-MRD code, which exists by Theorem \ref{MRD}. For each nonzero matrix $D=\left(
\begin{array}{c}
	\boldsymbol A \\
	\boldsymbol B \\
\end{array}
\right)\in\mathcal D_{a-1}$ with rank$(\boldsymbol A)=k-1$, define an invertible 
matrix of order $n$ as follows:
\begin{align*}
	M(D)=\left(
	\begin{array}{cc}
		\boldsymbol A & \boldsymbol I_{k-1} \\
		\boldsymbol B & \boldsymbol O \\
		\boldsymbol A & \boldsymbol O \\
	\end{array}
	\right).
\end{align*}
If $a>1$, let $\mathcal{D}_i$ be an  $[(n-(i+2)(k-1))\times (n-(i+2)(k-1)),n-(i+2)(k-1)]_q$-MRD code for $0\leq i\leq a-2$,  which exists by Theorem \ref{MRD}. For each nonzero matrix $G_i=\left(\begin{array}{c}
	\boldsymbol C \\
	\boldsymbol D \\
\end{array}
\right)\in \mathcal{D}_i$ with rank$(\boldsymbol C)=k-1$, define
\begin{align*}
	M(G_i)=
	\left(
	\begin{array}{cccc}
		\boldsymbol O & \boldsymbol I_{k-1} & \boldsymbol C & \boldsymbol O \\
		\boldsymbol O & \boldsymbol O & \boldsymbol D & \boldsymbol O \\
		\boldsymbol I_{i(k-1)} & \boldsymbol O & \boldsymbol O & \boldsymbol O \\	
		\boldsymbol O & \boldsymbol O & \boldsymbol O & \boldsymbol I_{k-1} \\
		\boldsymbol O & \boldsymbol O & \boldsymbol C & \boldsymbol O \\
	\end{array}
	\right).
\end{align*}
For the zero matrix $G_i=O\in \mathcal{D}_i$, define
\begin{align*}
	M(O)=
	\left( 
	\begin{array}{ccccc}
		\boldsymbol O & \boldsymbol I_{k-1} & \boldsymbol O  & \boldsymbol O& \boldsymbol O \\
		\boldsymbol O & \boldsymbol O & \boldsymbol O & \boldsymbol I_{n-(i+3)(k-1)} & \boldsymbol O \\
		\boldsymbol I_{i(k-1)} & \boldsymbol O & \boldsymbol O  & \boldsymbol O & \boldsymbol O \\
		\boldsymbol O & \boldsymbol O &  \boldsymbol O &\boldsymbol O & \boldsymbol I_{k-1} \\
		\boldsymbol O & \boldsymbol O &  \boldsymbol I_{k-1} & \boldsymbol O & \boldsymbol O \\
	\end{array}
	\right),
\end{align*}
Define two invertible matrices as follows:
\begin{align*}
	M(a)=\left(
	\begin{array}{ccc}
 	   \boldsymbol O & \boldsymbol O  & \boldsymbol I_{k-1} \\
 	  		\boldsymbol O &\boldsymbol I_{n-2(k-1)} & \boldsymbol O \\ 
		\boldsymbol I_{k-1} &\boldsymbol O & \boldsymbol O \\
	\end{array}
	\right),
\end{align*} 
\begin{align*}
	M(a+1)=\left(
	\begin{array}{ccc}
		\boldsymbol O  & \boldsymbol I_{k-1} & \boldsymbol O \\
		\hat{{\boldsymbol I}}_{n-2(k-1)} & \boldsymbol O & \boldsymbol O \\
	   \boldsymbol O &\boldsymbol O & \boldsymbol I_{k-1} \\
	\end{array}
	\right).
\end{align*}

Let 
$$\mathcal{M}=\{M(D)\colon D\in \mathcal{D}_{a-1}, D\neq O\}
 \cup\{M(G_i)\colon G_i\in \mathcal{D}_i,0\leq i\leq a-2\}
 \cup \{M(a),M(a+1)\}.$$

It is easy to see that each matrix of $\mathcal{M}$ is an invertible $n\times n$ matrix. Note that $\mathcal{M}=\{M(D)\colon D\in \mathcal{D}_{a-1}, D\neq O\}\cup \{M(a),M(a+1)\}$ if $a=1$. 

For each matrix $M\in \mathcal{M}$ and  $1\leq j\leq n$, let $\mathcal{F}(M_{[j]})=\text{rowspace}(M_{[j]})$,

$$\mathcal{F}(M)=(\mathcal{F}(M_{[1]}),\ldots,\mathcal{F}(M_{[k]}),
 \mathcal{F}(M_{[n-k+1]}),\ldots,\mathcal{F}(M_{[n-1]})),$$
where $M_{[j]}$ stands for  the top $j \times n$ submatrix of
$M$.

We claim that $\mathcal{C}=\{\mathcal{F}(M)\colon M\in \mathcal{M}\}$ is the desired flag code. 

{\bf Step 2.} We show that $\mathcal{C}$ is the set of flags of type $(1,2,\ldots,k,n-k+1,\ldots,n-1)$ and $|\mathcal{C}|=\frac{q^n-q^{k+r-1}}{q^{k-1}-1}+1$.

Since each matrix $M$ of $\mathcal{M}$ is an invertible $n\times n$ matrix, we have that $\dim (\mathcal{F}(M_{[j]}))=j$ for $j\in \{1,\ldots,k,n-k+1,\ldots,n-1\}$. Clearly, $\{0\}\subsetneq \mathcal{F}(M_{[1]})\subsetneq \mathcal{F}(M_{[2]})\subsetneq \cdots \subsetneq \mathcal{F}(M_{[k]}) \subsetneq  \mathcal{F}(M_{[n-k+1]}) \subsetneq \mathcal{F}(M_{[n-k+2]}) \subsetneq \cdots \subsetneq \mathcal{F}(M_{[n-1]})\subsetneq \mathbb{F}_q^n$. Therefore, $\mathcal{F}(M)$ is a flag of type $(1,\ldots,k,n-k+1,\ldots,n-1)$.
 
Clearly, $|\mathcal{C}|=|\mathcal{M}|$. Since $\mathcal D_{a-1}$ is an $[(n-k+1)\times(n-k+1),n-k+1]_q$-MRD code and each $\mathcal{D}_i$ is an $[(n-(i+2)(k-1))\times(n-(i+2)(k-1)),(n-(i+2)(k-1)]_q$-MRD code, by the construction of $\mathcal{M}$ we have 
\begin{align*}
	|\mathcal C|&=\sum\limits_{i=0}^{a-1}|\mathcal D_{i}|+1=q^{a(k-1)+r}+ \sum\limits_{i=1}^{a-1}{q^{i(k-1)+r}}+1\\ &=\sum\limits_{i=1}^{a}{q^{i(k-1)+r}}+1 =\frac{q^n-q^{k+r-1}}{q^{k-1}-1}+1.   
\end{align*}

{\bf Setp 3.} We prove that $\mathcal C$ is an almost optimum distance flag code. By Theorem \ref{AODFC-Proj} and Remark \ref{AODFCrek}, it suffices to show that subspace distance of the $t$-projected $\mathcal C^t$ of $\mathcal C$ are $2k-2$, where $t\in \{k-1,k,n-k+1\}$, $\mathcal C^t=\bigcup\limits_{i=0}^{a+1}\mathcal C_i^t$, $\mathcal{C}_{a-1}^t=\{\mathcal{F}(M(D)_{[t]})\colon D\in \mathcal{D}_{a-1}, D\neq \boldsymbol O\}$, $\mathcal{C}_i^t=\{\mathcal{F}(M(G)_{[t]})\colon G\in \mathcal{D}_i\}$, $\mathcal{C}_{a}^t=\{\mathcal{F}(M(a)_{[t]})\}$, and $\mathcal{C}_{a+1}^t=\{\mathcal{F}(M(a+1)_{[t]})\}$.
We distinguish three cases:

 \begin{itemize}
	\item[$(1)$] $t=k-1$.
	
	\item$(1.1)$ We shall prove that $d_S(\mathcal{C}_i^{k-1})\geq 2k-2$ for $0\leq i\leq a-1$. 
	
	For $0\leq i\leq a-2$, since $\mathcal{D}_{i}$ is an 
	$[(n-(i+2)(k-1))\times (n-(i+2)(k-1)),(n-(i+2)(k-1))]_q$-MRD code,
	$\{G_{[k-1]}\colon G\in \mathcal{D}_{i}\}$ is a $[(k-1)\times (n-(i+2)(k-1)),k-1]_q$-MRD code. For any two distinct subspaces $\mathcal U=\text{rowspace}(O_{(k-1)\times i(k-1)},I_{k-1}, G_{[k-1]},O)$, $\mathcal V=\text{rowspace}(O_{(k-1)\times i(k-1)},I_{k-1}, G'_{[k-1]},O) \in \mathcal C_{i}^{k-1}$, by Lemma \ref{lem:efc-1} we have $d_S(\mathcal U,\mathcal V)=2d_R(G_{[k-1]},G'_{[k-1]})= 2k-2$. 
	
Since $\mathcal{D}_{a-1}$ is an 
		$[(n-k+1)\times (n-k+1),n-k+1]_q$-MRD code,
		$\{D_{[k]}\colon D\in \mathcal{D}_{a-1}\}$ is a $[(k-1)\times (n-k+1),k-1]_q$-MRD code. For any two distinct subspaces $\mathcal U=\text{rowspace}(D_{[k-1]},I_{k-1},O)$, $\mathcal V=\text{rowspace}(D'_{[k-1]},I_{k-1},O) \in \mathcal C_{a-1}^{k-1}$, by Lemma \ref{lem:iefc-1} we have $d_S(\mathcal U,\mathcal V)=2d_R(D_{[k-1]},D'_{[k-1]})= 2k-2$. 
	
	\item$(1.2)$  We shall prove that $d_S(\mathcal U,\mathcal V)\geq 2k-2$ for any two subspaces $\mathcal U\in \mathcal C^{k-1}_{a}, \mathcal V \in \mathcal C_i^{k-1}$, where $0\leq i\leq a-1$. 
	
	Since the indentifying vector of $\mathcal U$ is  $(\underbrace{0\cdots0}_{n-k+1}\underbrace{1\cdots1}_{k-1})$ and the indentifying vector of $\mathcal V$ is $(\underbrace{0\cdots0}_{i(k-1)}\underbrace{1\cdots1}_{k-1}\underbrace{0\cdots0}_{n-(i+1)(k-1)})$ if $i<a-1$ or $(\underbrace{\alpha }_{n-k+1}\underbrace{0\cdots0}_{k-1})$ with $wt(\alpha)=k-1$ if $i=a-1$,  we have that $d_S(\mathcal U,\mathcal V)\geq 2k-2$ by Lemma
	\ref{lem:efc-2}.
	
	\item$(1.3)$  We shall prove that $d_S(\mathcal U,\mathcal V)\geq 2k-2$ for any two subspaces $\mathcal U\in \mathcal C^{k-1}_{a+1}, \mathcal V \in \mathcal C_i^{k-1}$, where $0\leq i\leq a$.
	
	Since both the identifying vector  and inverse identifying vector of $\mathcal U$ are  $(\underbrace{0\cdots0}_{n-2(k-1)}\\\underbrace{1\cdots1}_{k-1}\underbrace{0\cdots0}_{k-1})$, the identifying vector of $\mathcal V$ is $(\underbrace{0\cdots0}_{i(k-1)}\underbrace{1\cdots1}_{k-1}\underbrace{0\cdots0}_{n-(i+1)(k-1)})$ when $i\leq a-2$, the inverse identifying vector of $\mathcal V$ is $(\underbrace{0\cdots0}_{n-k+1}\underbrace{1\cdots1}_{k-1})$ when $i= a-1$, and both the identifying vector and  inverse identifying vector of $\mathcal V$ are $(\underbrace{0\cdots0}_{n-k+1}\underbrace{1\cdots1}_{k-1})$ when $i= a$, we have that $d_S(\mathcal U,\mathcal V)\geq 2k-2$ by Lemmas    \ref{lem:efc-2} and  \ref{lem:iefc}.
	
	\item$(1.4)$  We shall prove that $d_S(\mathcal U,\mathcal V)\geq 2k-2$ for any two subspaces $\mathcal U\in \mathcal C^{k-1}_{i}, \mathcal V \in \mathcal C_j^{k-1}$, where $0\leq i< j\leq a-2$.
	
	Since the indentifying vectors of $\mathcal U$ and $\mathcal V$ are $(\underbrace{0\cdots0}_{i(k-1)}\underbrace{1\cdots1}_{k-1}\underbrace{0\cdots0}_{n-(i+1)(k-1)})$ and $(\underbrace{0\cdots0}_{j(k-1)}\\\underbrace{1\cdots1}_{k-1}\underbrace{0\cdots0}_{n-(j+1)(k-1)})$ respectively,
	we have that $d_S(\mathcal U,\mathcal V)\geq 2k-2$ by Lemma
	\ref{lem:efc-2}.
	\item$(1.5)$  We shall prove that $d_S(\mathcal U,\mathcal V)\geq 2k-2$ for any two subspaces $\mathcal U\in \mathcal C^{k-1}_{a-1}, \mathcal V \in \mathcal C^{k-1}_{i}$, where $0\leq i\leq a-2$. 
	
	Since the inverse indentifying vectors of $\mathcal U$ and $\mathcal V$ are $(\underbrace{0\cdots0}_{n-k+1}\underbrace{1\cdots1}_{k-1})$ and $(\underbrace{\alpha}_{n-k+1}\underbrace{0\cdots0}_{k-1})$ with $wt(\alpha)=k-1$ respectively,
we have that $d_S(\mathcal U,\mathcal V)\geq 2k-2$ by Lemma
\ref{lem:iefc}.

\item[$(2)$]  $t=k$.
	
	\item$(2.1)$  We shall prove that $d_S(\mathcal{C}_i^{k})\geq 2k-2$ for $0\leq i\leq a-1$.
	
	\item $(2.1.1)$ For any two distinct subspaces $\mathcal U, \mathcal V \in \mathcal C_{a-1}^{k}$, we have 
	\begin{align*}
		d_S(\mathcal U,\mathcal V)&=2{\rm rank}\left(
		\begin{array}{ccccc}
			\boldsymbol A & \boldsymbol I_k \\
			\boldsymbol B_{[1]} & \boldsymbol O\\
			\boldsymbol A' & \boldsymbol I_k \\
			\boldsymbol B'_{[1]} & \boldsymbol O\\
		\end{array}
		\right)-2k\\
		&=2{\rm rank}\left(
		\begin{array}{ccccc}
			\boldsymbol A-\boldsymbol A' & \boldsymbol O \\
			\boldsymbol B_{[1]}-\boldsymbol B'_{[1]} & \boldsymbol O\\
			\boldsymbol A' & \boldsymbol I_k \\
		\end{array}
		\right)-2k\\&=2(k-1+k-k)=2k-2,
	\end{align*}
	where the last equality holds because $\left(
	\begin{array}{c}
		\boldsymbol A \\
		\boldsymbol B \\
	\end{array}
	\right),\left (\begin{array}{c}
		\boldsymbol A' \\
		\boldsymbol B' \\
	\end{array}
	\right)$ are two distinct nonzero codewords of $[(n-k+1)\times (n-k+1),n-k+1]_q$-MRD code $\mathcal D$.

	\item$(2.1.2)$ For any two distinct subspaces $\mathcal U=\mathcal{F}(M(G)_{[k]}), \mathcal V=\mathcal{F}(M(G')_{[k]}) \in \mathcal C_i^{k}$ with $G\neq O, G'\neq O$ and $0\leq i<a-1$, similar to (2.1.1), it is easy to verify that $d_S(\mathcal U, \mathcal V )\geq 2k-2$.
	\item$(2.1.3)$ For any subspace $\mathcal U=\mathcal{F}(M(G)_{[k]}), \mathcal V=\mathcal{F}(M(O)_{[k]}) \in \mathcal C_i^{k}$ with $G\neq O$ and $0\leq i<a-1$, we have
	{\begin{align*}
			d_S(\mathcal U,\mathcal V)&=2{\rm rank}\left(
			\begin{array}{cccccc}
				\boldsymbol O& \boldsymbol I_{k-1} & \boldsymbol C & \boldsymbol O & \boldsymbol O\\
				\boldsymbol O &\boldsymbol O& \boldsymbol D_{[1]} & \boldsymbol O & \boldsymbol O\\
				\hdashline[2pt/2pt]
				\boldsymbol O& \boldsymbol I_{k-1} & \boldsymbol O & \boldsymbol O& \boldsymbol O \\
			\end{array}
			\right)-2k\\&=2{\rm rank}\left(
			\begin{array}{cccccc}
			\boldsymbol O& \boldsymbol I_{k-1} & \boldsymbol C & \boldsymbol O & \boldsymbol O\\
				\boldsymbol O &\boldsymbol O& \boldsymbol D_{[1]} & \boldsymbol O & \boldsymbol O\\
				\hdashline[2pt/2pt]
				\boldsymbol O& \boldsymbol O & {\boldsymbol -C} & \boldsymbol O& \boldsymbol O \\
			\end{array}
			\right)-2k\\&=2(k-1+k-k)=2k-2.
	\end{align*}}

	\item$(2.2)$ We shall prove that $d_S(\mathcal U,\mathcal V)\geq 2k-2$
	for any two subspaces $\mathcal U\in \mathcal C^{k}_{a}, \mathcal V \in \mathcal C_i^{k}$, where $0\leq i\leq a-1$.
	
	Since the indentifying vector of $\mathcal U$ is  $(\underbrace{0\cdots0}_{k-1}1\underbrace{0\cdots0}_{n-2k+1}\underbrace{1\cdots1}_{k-1})$, the indentifying vector of $\mathcal V$ is of the form $(\underbrace{0\cdots0}_{i(k-1)}\underbrace{1\cdots1}_{k-1}\underbrace{\alpha_2}_{n-(i+2)(k-1)}\underbrace{0\cdots0}_{k-1})$ for some $\alpha_2\in \mathbb F_2^{n-(i+2)(k-1)}$ of Hamming weight $1$ when $i\leq a-2$, and the indentifying vector of $\mathcal V$ is of the form $(\underbrace{\alpha_{2}}_{n-k+1}\underbrace{0\cdots0}_{k-1})$  for some $\alpha_2\in \mathbb F_2^{n-k+1}$ of Hamming weight $k$ when $i=a-1$, we have that $d_S(\mathcal U,\mathcal V)\geq 2k-2$ by Lemma
	\ref{lem:efc-2}.
	
	\item$(2.3)$ We shall prove that $d_S(\mathcal U,\mathcal V)\geq 2k-2$ for any two subspaces $\mathcal U\in \mathcal C^{k}_{a+1}, \mathcal V \in \mathcal C_i^{k}$, where $0\leq i\leq a$.
	
	Since the identifying vector and inverse indentifying vector of $\mathcal U$ are  $(\underbrace{0\cdots0}_{n-2k+1}\underbrace{1\cdots1}_{k}\\\underbrace{0\cdots0}_{k-1})$, the  indentifying vector of $\mathcal V$ is of the form $(\underbrace{0\cdots0}_{i(k-1)}\underbrace{1\cdots1}_{k-1}\underbrace{\alpha_{2}}_{n-(i+2)(k-1)}\underbrace{1\cdots1}_{k-1})$ for some $\alpha_2\in \mathbb F_2^{n-(i+2)(k-1)}$ of Hamming weight $1$ when $i\leq a-2$,   the inverse  indentifying vector of $\mathcal V$ is of the form $(\underbrace{\alpha_{2}}_{n-k+1}\underbrace{1\cdots1}_{k-1})$ for some $\alpha_2\in \mathbb F_2^{n-k+1}$ of Hamming weight $1$ when $i=a-1$, the indentifying vector of $\mathcal V$ is $(\underbrace{0\cdots0}_{k-1}1\underbrace{0\cdots0}_{n-2k+1}\underbrace{1\cdots1}_{k-1})$ when $i=a$, we have that $d_S(\mathcal U,\mathcal V)\geq 2k-2$ by Lemma
	\ref{lem:iefc}.
	\item$(2.4)$ We shall prove that $d_S(\mathcal U,\mathcal V)\geq 2k-2$ for any two subspaces $\mathcal U\in \mathcal C^{k}_{a-1}$, $\mathcal V \in \mathcal C^{k}_{i}$, where $0\leq i\leq a-2$.
	
Since the inverse indentifying vector of $\mathcal U$ is of the form  $(\underbrace{\alpha}_{n-k+1}\underbrace{1\cdots1}_{k-1})$ for some $\alpha\in \mathbb F_2^{n-k+1}$ of Hamming weight $1$, the  indentifying vector of $\mathcal V$ is of the form $(\underbrace{0\cdots0}_{(i+1)(k-1)}\underbrace{\alpha_{2}}_{n-(i+2)(k-1)}\underbrace{0\cdots0}_{k-1})$ for some $\alpha_2\in \mathbb F_2^{n-(i+2)(k-1)}$ of Hamming weight $1$, we have that $d_S(\mathcal U,\mathcal V)\geq 2k-2$ by Lemma
	\ref{lem:iefc}.
	
	\item$(2.5)$ We shall prove that $d_S(\mathcal U,\mathcal V)\geq 2k-2$ for any two subspaces $\mathcal U\in \mathcal C^{n-k}_i$, $\mathcal V \in \mathcal C^{n-k}_{j}$, where $0\leq i<j\leq a-2$.

	Since the indentifying vector of $\mathcal U$ is of the form  $(\underbrace{0\cdots0}_{i(k-1)}\underbrace{1\cdots1}_{k-1}\underbrace{\alpha}_{n-(i+2)(k-1)}\underbrace{1\cdots1}_{k-1})$ for some $\alpha\in \mathbb F_2^{n-(i+2)(k-1)}$ of Hamming weight $1$, the  indentifying vector of $\mathcal V$ is of the form $(\underbrace{0\cdots0}_{j(k-1)}\underbrace{1\cdots1}_{k-1}\underbrace{\beta}_{n-(j+2)(k-1)}\underbrace{0\cdots0}_{k-1})$ for some $\beta\in \mathbb F_2^{n-(j+2)(k-1)}$ of Hamming weight $1$, we have that $d_S(\mathcal U,\mathcal V)\geq 2k-2$ by Lemma 
	\ref{lem:efc-2}.

	\item[$(3)$]  $t=n-k+1$.
	
	\item$(3.1)$  We shall prove that $d_S(\mathcal{C}_i^t)\geq 2k-2$ for $0\leq i\leq a-1$.
	
	\item $(3.1.1)$ For any two distinct subspaces $\mathcal U, \mathcal V \in \mathcal C_{a-1}^{n-k+1}$, we have 
	\begin{align*}
		d_S(\mathcal U,\mathcal V)&=2{\rm rank}\left(
		\begin{array}{ccccc}
			\boldsymbol A & \boldsymbol I_{k-1} \\
			\boldsymbol B & \boldsymbol O\\
			\boldsymbol A' & \boldsymbol I_{k-1} \\
			\boldsymbol B' & \boldsymbol O\\
		\end{array}
		\right)-2(n-k+1)\\
		&=2{\rm rank}\left(
		\begin{array}{ccccc}
			\boldsymbol A-\boldsymbol A' & \boldsymbol O \\
			\boldsymbol B-\boldsymbol B' & \boldsymbol O\\
			\boldsymbol A' & \boldsymbol I_{k-1} \\
		\end{array}
		\right)-2(n-k+1)\\&=2(n-k+1+k-1-n+k-1)=2k-2,
	\end{align*}
	where the last equality holds because $\left(
	\begin{array}{c}
		\boldsymbol A \\
		\boldsymbol B \\
	\end{array}
	\right),\left (\begin{array}{c}
		\boldsymbol A' \\
		\boldsymbol B' \\
	\end{array}
	\right)$ are two distinct nonzero codewords of $[(n-k+1)\times (n-k+1),n-k+1]_q$-MRD code $\mathcal D$.

	\item$(3.1.2)$ For any two distinct subspaces $\mathcal U=\mathcal{F}(M(G)_{[n-k+1]}), \mathcal V=\mathcal{F}(M(G')_{[n-k+1]}) \in \mathcal C_i^{n-k+1}$ with $G\neq O, G'\neq O$ and $0\leq i<a-1$, similar to (2.1.1), it is easy to verify that $d_S(\mathcal U, \mathcal V )\geq 2k-2$.
	\item$(3.1.3)$ For any subspace $\mathcal U=\mathcal{F}(M(G)_{[n-k+1]}), \mathcal V=\mathcal{F}(M(O)_{[n-k+1]}) \in \mathcal C_i^{n-k+1}$ with $G\neq O$ and $0\leq i<a-1$, we have
	{\begin{align*}
			d_S(\mathcal U,\mathcal V)&=2{\rm rank}\left(
			\begin{array}{cccccc}
				\boldsymbol O&\boldsymbol I_{k-1} & \boldsymbol C & \boldsymbol O  \\
				\boldsymbol O & \boldsymbol O& \boldsymbol D & \boldsymbol O \\
				\boldsymbol I_{i(k-1)} & \boldsymbol O & \boldsymbol O & \boldsymbol O \\
				\boldsymbol O & \boldsymbol O & \boldsymbol O & \boldsymbol I_{k-1} \\
				\hdashline[2pt/2pt]
                          \boldsymbol O & \boxed{\boldsymbol I_{k-1}} & \boldsymbol O & \boldsymbol O \\
				\boldsymbol I_{i(k-1)} & \boldsymbol O & \boldsymbol O & \boldsymbol O \\
				\boldsymbol O & \boldsymbol O & * & * \\
			\end{array}
			\right)-2(n-k+1)\\&=2{\rm rank}\left(
			\begin{array}{cccc}
\boldsymbol O&\boldsymbol O & \boldsymbol C & \boldsymbol O  \\
				\boldsymbol O & \boldsymbol O& \boldsymbol D & \boldsymbol O \\
				\boldsymbol I_{i(k-1)} & \boldsymbol O & \boldsymbol O & \boldsymbol O \\
				\boldsymbol O & \boldsymbol O & \boldsymbol O & \boldsymbol I_{k-1} \\
				\hdashline[2pt/2pt]
                          \boldsymbol O &\boldsymbol I_{k-1}& \boldsymbol O & \boldsymbol O \\
			\end{array}
			\right)-2(n-k+1)
                  \\&\geq 2{\rm rank}\left(
			\begin{array}{ccc;{2pt/2pt}c}
				\boldsymbol I_{i(k-1)} & \boldsymbol O & \boldsymbol O & \boldsymbol O \\
				\boldsymbol O & \boldsymbol I_{k-1} & \boldsymbol O & \boldsymbol O \\
				\boldsymbol O & \boldsymbol O  & \boldsymbol C & \boldsymbol O \\
				\boldsymbol O & \boldsymbol O & \boldsymbol D & \boldsymbol O \\\hdashline[2pt/2pt]
				\boldsymbol O & \boldsymbol O & \boldsymbol O & \boldsymbol I_{k-1} \\
			\end{array}
			\right)-2(n-k+1)\\&=2(i(k-1)+k-1+n-(i+2)(k-1)+k-1)-2(n-k+1)\\&=2k-2,
	\end{align*}}
where $*$ stands for a submatrix.

	\item$(3.2)$ We shall prove that $d_S(\mathcal U,\mathcal V)\geq 2k$
	for any two subspaces $\mathcal U\in \mathcal C^{n-k+1}_{a}, \mathcal V \in \mathcal C_i^{n-k+1}$, where $0\leq i\leq a-1$.
	
	Since the indentifying vector of $\mathcal U$ is  $(\underbrace{0\cdots0}_{k-1}\underbrace{1\cdots1}_{n-k+1})$, the indentifying vector of $\mathcal V$ is of the form $(\underbrace{1\cdots1}_{(i+1)(k-1)}\underbrace{\alpha_2}_{n-(i+2)(k-1)}\underbrace{1\cdots1}_{k-1})$ for some $\alpha_2\in \mathbb F_2^{n-(i+2)(k-1)}$ of Hamming weight $n-(i+3)(k-1)$ when $i\leq a-2$, and the indentifying vector of $\mathcal V$ is of the form $(\underbrace{1\cdots1}_{n-k+1}\underbrace{0\cdots0}_{k-1})$ when $i=a-1$, we have that $d_S(\mathcal U,\mathcal V)\geq 2k-2$ by Lemma
	\ref{lem:efc-2}.
	
	\item$(3.3)$ We shall prove that $d_S(\mathcal U,\mathcal V)\geq 2k-2$ for any two subspaces $\mathcal U\in \mathcal C^{n-k+1}_{a+1}, \mathcal V \in \mathcal C_i^{n-k+1}$, where $0\leq i\leq a$.
	
	Since the inverse indentifying vector of $\mathcal U$ is  $(\underbrace{1\cdots1}_{n-k+1}\underbrace{0\cdots0}_{k-1})$ and the inverse indentifying vector of $\mathcal V$ is of the form $(\underbrace{\alpha}_{n-k+1}\underbrace{1\cdots1}_{k-1})$ for some $\alpha\in \mathbb F_2^{n-k+1}$ of Hamming weight $n-2(k-1)$,  we have that $d_S(\mathcal U,\mathcal V)\geq 2k-2$ by Lemma
	\ref{lem:iefc}.

	\item$(3.4)$ We shall prove that $d_S(\mathcal U,\mathcal V)\geq 2k-2$ for any two subspaces $\mathcal U\in \mathcal C^{n-k+1}_{a-1}$, $\mathcal V \in \mathcal C^{n-k}_{i}$, where $0\leq i\leq a-2$.
	
	\item$(3.4.1)$  Let $\mathcal U=\mathcal{F}(M(D)_{[n-k+1]})\in \mathcal C^{n-k+1}_{a-1}, \mathcal V=\mathcal{F}(M(G)_{[n-k+1]}) \in \mathcal C^{n-k+1}_{i}$ with $D$ and $G$ nonzero matrices. We have
{ \small\begin{align*}
			d_S(\mathcal U,\mathcal V)&=2{\rm rank}\left(
			\begin{array}{c;{2pt/2pt}c}
				\boldsymbol A & \boxed{\boldsymbol I_{k-1}} \\
				\boldsymbol B& \boldsymbol O \\
				\hdashline[2pt/2pt]
				\begin{array}{ccc}
					\boldsymbol I_{i(k-1)} & \boldsymbol O & \boldsymbol O  \\
					\boldsymbol O & \boldsymbol I_{k-1} & \boldsymbol C  \\
					\boldsymbol O & \boldsymbol O &\boldsymbol  D\\
					\boldsymbol O & \boldsymbol O & \boldsymbol O  \\
				\end{array}
				& \begin{array}{c}
					\boldsymbol O \\
					\boldsymbol O \\
					\boldsymbol O \\
					\boxed{\boldsymbol I_{k-1}} \\
				\end{array}
				\\
			\end{array}
			\right)-2(n-k+1)
\\&=2{\rm rank}\left(
			\begin{array}{c;{2pt/2pt}c}
				\boldsymbol A & \boldsymbol O \\
				\boldsymbol B& \boldsymbol O \\
				\hdashline[2pt/2pt]
				\begin{array}{ccc}
					\boldsymbol I_{i(k-1)} & \boldsymbol O & \boldsymbol O  \\
					\boldsymbol O & \boldsymbol I_{k-1} & \boldsymbol C  \\
					\boldsymbol O & \boldsymbol O &\boldsymbol  D\\
					\boldsymbol O & \boldsymbol O & \boldsymbol O  \\
				\end{array}
				& \begin{array}{c}
					\boldsymbol O \\
					\boldsymbol O \\
					\boldsymbol O \\
					\boxed{\boldsymbol I_{k-1}} \\
				\end{array}
				\\
			\end{array}
			\right)-2(n-k+1)\\
&\geq 2[{\rm rank}\left(
			\begin{array}{c}
				\boldsymbol A  \\
				\boldsymbol B  \\
			\end{array}
			\right)+k-1]-2(n-k+1)\\&=2k-2.
	\end{align*}}
	
	\item$(3.4.2)$  Let $\mathcal U=\mathcal{F}(M(D)_{[n-k+1]})\in \mathcal C^{n-k+1}_{a-1}, \mathcal V=\mathcal{F}(M(O)_{[n-k+1]}) \in \mathcal C^{n-k+1}_{i}$ with  $D$ a nonzero matrix. We have
	{ \small\begin{align*}
			d_S(\mathcal U,\mathcal V)&=2{\rm rank}\left(
			\begin{array}{c;{2pt/2pt}c}
				\boldsymbol A & \boxed{\boldsymbol I_{k-1}} \\
				\boldsymbol B& \boldsymbol O \\
				\hdashline[2pt/2pt]
				\begin{array}{ccc}
					\boldsymbol I_{i(k-1)} & \boldsymbol O & \boldsymbol O  \\
					\boldsymbol O & \boldsymbol I_{k-1} & \boldsymbol O  \\
					\boldsymbol O & \boldsymbol O & *\\
					\boldsymbol O & \boldsymbol O & \boldsymbol O  \\
				\end{array}
				& \begin{array}{c}
					\boldsymbol O \\
					\boldsymbol O \\
					\boldsymbol O \\
					\boxed{\boldsymbol I_{k-1}} \\
				\end{array}
				\\
			\end{array}
			\right)-2(n-k+1)\\&=2{\rm rank}\left(
			\begin{array}{c;{2pt/2pt}c}
				\boldsymbol A & \boldsymbol O \\
				\boldsymbol B& \boldsymbol O \\
				\hdashline[2pt/2pt]
				\begin{array}{ccc}
					\boldsymbol I_{i(k-1)} & \boldsymbol O & \boldsymbol O  \\
					\boldsymbol O & \boldsymbol I_{k-1} & \boldsymbol O  \\
					\boldsymbol O & \boldsymbol O & * \\
					\boldsymbol O & \boldsymbol O & \boldsymbol O  \\
				\end{array}
				& \begin{array}{c}
					\boldsymbol O \\
					\boldsymbol O \\
					\boldsymbol O \\
					\boxed{\boldsymbol I_{k-1}} \\
				\end{array}
				\\
			\end{array}
			\right)-2(n-k+1)\\&
			\geq 2[{\rm rank}\left(
			\begin{array}{c}
				\boldsymbol A  \\
				\boldsymbol B  \\
			\end{array}
			\right)+k-1]-2(n-k+1)\\&=2k-2.
	\end{align*}}
	
	\item$(3.5)$ We shall prove that $d_S(\mathcal U,\mathcal V)\geq 2k-2$ for any two subspaces $\mathcal U\in \mathcal C^{n-k}_i$, $\mathcal V \in \mathcal C^{n-k}_{j}$, where $0\leq i<j\leq a-2$.
	
	\item$(3.5.1)$
	For any two subspaces $\mathcal U=\mathcal{F}(M(G)_{[n-k+1]})\in \mathcal C^{n-k+1}_{i}, \mathcal V=\mathcal{F}(M(G')_{[n-k+1]}) \in \mathcal C^{n-k}_{j}$ with $G\in \mathcal{D}_i$, $G'\in \mathcal{D}_{j}$ and $G\neq O$, we have
	{\begin{align*}
			d_S(\mathcal U,\mathcal V)&=2{\rm rank}\left(
			\begin{array}{cccc}
				\boldsymbol I_{i(k-1)} & \boldsymbol O & \boldsymbol O & \boldsymbol O \\
				\boldsymbol O & \boxed{\boldsymbol I_{k-1}} & \boldsymbol C & \boldsymbol O \\
				\boldsymbol O & \boldsymbol O & \boldsymbol D & \boldsymbol O \\
				\boldsymbol O & \boldsymbol O & \boldsymbol O & \boldsymbol I_{k-1} \\
				\hdashline[2pt/2pt]
				\boldsymbol I_{i(k-1)} & \boldsymbol O & \boldsymbol O & \boldsymbol O \\
				\boldsymbol O & \boxed{\boldsymbol I_{k-1}} & \boldsymbol O & \boldsymbol O \\
				\boldsymbol O & \boldsymbol O & * & * \\
			\end{array}
			\right)-2(n-k+1)\\&=2{\rm rank}\left(
			\begin{array}{cccc}
				\boldsymbol I_{i(k-1)} & \boldsymbol O & \boldsymbol O & \boldsymbol O \\
				\boldsymbol O & \boldsymbol O  & \boldsymbol C & \boldsymbol O \\
				\boldsymbol O & \boldsymbol O & \boldsymbol D & \boldsymbol O \\
				\boldsymbol O & \boldsymbol O & \boldsymbol O & \boldsymbol I_{k-1} \\
				\hdashline[2pt/2pt]
				\boldsymbol I_{i(k-1)} & \boldsymbol O & \boldsymbol O & \boldsymbol O \\
				\boldsymbol O & \boldsymbol I_{k-1} & \boldsymbol O & \boldsymbol O \\
				\boldsymbol O & \boldsymbol O & * & * \\
			\end{array}
			\right)-2(n-k+1)
             \\&\geq 2{\rm rank}\left(
			\begin{array}{ccc;{2pt/2pt}c}
				\boldsymbol I_{i(k-1)} & \boldsymbol O & \boldsymbol O & \boldsymbol O \\
				\boldsymbol O & \boldsymbol I_{k-1} & \boldsymbol O & \boldsymbol O \\
				\boldsymbol O & \boldsymbol O  & \boldsymbol C & \boldsymbol O \\
				\boldsymbol O & \boldsymbol O & \boldsymbol D & \boldsymbol O \\\hdashline[2pt/2pt]
				\boldsymbol O & \boldsymbol O & \boldsymbol O & \boldsymbol I_{k-1} \\
			\end{array}
			\right)-2(n-k+1)
\\&=2(i(k-1)+k-1+n-(i+2)(k-1)+k-1)-2(n-k+1)\\&=2k-2.
	\end{align*}}
	\item$(3.5.2)$ For  any subspace $\mathcal U=\mathcal{F}(M(O)_{[n-k+1]})\in \mathcal C^{n-k+1}_i$, $\mathcal V=\mathcal{F}(M(G)_{[n-k+1]}) \in \mathcal C^{n-k+1}_{j}$ with $G\in \mathcal{D}_{j}$, we have
	{\begin{align*}
			d_S(\mathcal U,\mathcal V)&=2{\rm rank}\left(
			\begin{array}{cccc}
				\boldsymbol I_{(i+1)(k-1)} &   \boldsymbol O & \boldsymbol O \\
				\boldsymbol O &  \boldsymbol O& \boldsymbol I_{n-(i+2)(k-1)} \\
				\boldsymbol O & \boldsymbol *  & \boldsymbol O \\
				\boldsymbol I_{(i+1)(k-1)} & \boldsymbol O  & \boldsymbol O \\
				\boldsymbol O &\boldsymbol I_{k-1} & \boldsymbol * &  \\
				\boldsymbol O & \boldsymbol O &\boldsymbol *   \\
			\end{array}
			\right)-2(n-k+1)\\&\geq 2{\rm rank}\left(
			\begin{array}{ccccc}
				\boldsymbol I_{(i+1)(k-1)} & \boldsymbol O & \boldsymbol O \\
				\boldsymbol O & \boldsymbol I_{k-1} & \boldsymbol O  \\
			    \boldsymbol O & \boldsymbol O  & \boldsymbol I_{n-(i+2)(k-1)}
			\end{array}
			\right)-2(n-k+1)\\&=2(i(k-1)+2(k-1)+n-(i+3)(k-1)+k-1)-2(n-k+1)\\&=2k-2.
	\end{align*}}
\end{itemize}
\end{proof}

\begin{corollary}\label{rek2}
Let $n,k$ be positive integers with $n\geq 2k$ . If $k>[_1^r]_q+1$, then
\begin{align*}
 A_q^{*}(n,\{1,2,\ldots,k,n-k+1,\ldots,n-1\})=\frac{q^n-q^{k+r-1}}{q^{k-1}-1}+1,
\end{align*}
where $r \equiv n \pmod{(k-1)}$ and $0\leq r<k-1.$
\end{corollary}

\begin{proof}
    Apply Theorems \ref{flagupper2} and \ref{conc-2}.
\end{proof}

\section{Concluding remarks}
In this paper, we construct $(n,\mathcal A)_q$-ODFCs($(n,\{1,2,\ldots,k,n-k+1,\dots,n-1\})_q$-AODFCs) by MRD codes, where $\mathcal A\subseteq \{1,2,\ldots,k,n-k,\ldots,n-1\}$ and also show that for $k>\begin{bmatrix}
 r \\
 1 \\
\end{bmatrix}_q$($k>\begin{bmatrix}
 r' \\
 1 \\
\end{bmatrix}_q+1$), the ODFCs(AODFCs) from Theorem \ref{conc-1}(Theorem \ref{conc-2}) are optimal $(n,\mathcal A)_q$-ODFCs($(n,\{1,2,\ldots,k,n-k+1,\dots,n-1\})_q$-AODFCs),  where $n\equiv r\pmod k$, $n \equiv r' \pmod{(k-1)}$. However, when $k\leq \begin{bmatrix}
 r \\
 1 \\
\end{bmatrix}_q$($k\leq \begin{bmatrix}
 r' \\
 1 \\
\end{bmatrix}_q+1$), the gap between the upper and lower bounds of ODFCs(AODFCs) still relies on $q$.

Recently, Fourier Nebe \cite{FN2023}  highlight the advantages of using degenerate flags instead of flags in network coding. The information set is a metric affine space isometric to the space of upper triangular
matrices endowed with the flag-rank metric. Alfarano, Neri and Zullo provided a Singleton-like bound which relates the parameters of a flag-rank metric code. They also provided several constructions
of maximum flag-rank distance codes \cite{ANZ}.  More constructions for maximum flag-rank distance codes are desired. 
	
\bigskip
\noindent {\bf Data availability}

No data was used for the research described in this article.
	
\bigskip
\noindent {\bf Conflict of Interest}

The authors declare that they have no conflicts of interest.

\end{document}